\documentclass[11pt]{article} 

\usepackage[round,authoryear]{natbib}  
\usepackage[nottoc, notlof, notlot]{tocbibind}  

\usepackage{lmodern}  
\usepackage[french,english]{babel} 
\usepackage[utf8]{inputenc} 
\usepackage[T1]{fontenc}  
\usepackage{geometry}  

\usepackage[explicit]{titlesec} 
\usepackage{tocloft} 
\usepackage{titletoc} 
\usepackage[normalem]{ulem} 
\usepackage{authblk}

\usepackage{array}  
\usepackage{booktabs}  
\usepackage{tabularx}  
\usepackage{tabulary}  
\usepackage{multicol}  
\usepackage{enumitem}  

\usepackage{graphicx}  
\usepackage{subfig}  
\usepackage{float}  
\usepackage{caption}  
\usepackage[dvipsnames]{xcolor}  
\usepackage{tikz}  
\usetikzlibrary{arrows.meta, shapes, positioning}  
\usepackage{pdfpages} 

\usepackage{amsmath, amsthm, amssymb, amsfonts, mathtools}  
\usepackage{stmaryrd}  
\usepackage{bbm, dsfont}  

\usepackage{algorithm}  
\usepackage{algorithmic}  

\usepackage{siunitx}  

\usepackage{hyperref}  
\hypersetup{colorlinks=true, citecolor=ForestGreen, urlcolor=ForestGreen, allcolors=ForestGreen}

\usepackage{graphicx} 
\usepackage{amsmath} 
\usepackage{float}
\usepackage{tikz}
\usetikzlibrary{arrows.meta,positioning,calc}
\usepackage{array}
\usepackage{booktabs}

\newcommand{\trajrowcolor}[2]{%
  \foreach \x [count=\i from 0] in {#2} {%
    \ifnum\x=1
      \filldraw[#1] (\i*0.32,0) rectangle ++(0.28,0.28);
    \else
      \draw[black] (\i*0.32,0) rectangle ++(0.28,0.28);
    \fi
  }%
}

\newcommand{\trajaxis}{%
\makebox[0pt][r]{\scriptsize$k=$\hspace{0.1cm}}%
\begin{tikzpicture}[baseline=-0.1cm]
  \foreach \i in {0,...,9} {%
    \node[font=\scriptsize] at (\i*0.32+0.13,0) {$\i$};
  }%
\end{tikzpicture}
}

\newcommand{\trajlinecolor}[3]{%
#1 &
\begin{tikzpicture}[baseline=-0.1cm]
  \trajrowcolor{#2}{#3}
\end{tikzpicture}
\\[-0.03cm]
}

\newcommand{\trajlineDOne}[2]{%
  \trajlinecolor{#1}{blue}{#2}
}
\newcommand{\trajlineDOneFaded}[2]{%
  \trajlinecolor{#1}{blue!35}{#2}
}

\newcommand{\trajlineDTwo}[2]{%
  \trajlinecolor{#1}{green!60!black}{#2}
}
\newcommand{\trajlineDTwoFaded}[2]{%
  \trajlinecolor{#1}{green!35}{#2}
}

\newcommand{\DOnePct}[1]{%
  \makebox[1.45cm][l]{\textcolor{blue!70!black}{$\Omega_0$}\hfill{\scriptsize(#1)}}%
}

\newcommand{\DTwoPct}[1]{%
  \makebox[1.45cm][l]{\textcolor{green!50!black}{$\Omega_1$}\hfill{\scriptsize(#1)}}%
}

\def\cO{{\mtc{O}}}

\renewcommand{\emph}[1]{{\it #1}}

\newcommand{\R}{\mathbb{R}}

\newcommand{\ind}[1]{{\mathbf{1}\{#1\}}}

\newcommand{\paren}[2][a]{%
\IfEqCase{#1}{%
{a}{\left(#2\right)}%
{0}{(#2)}%
{1}{\big(#2\big)}%
{2}{\Big(#2\Big)}%
{3}{\bigg(#2\bigg)}%
{4}{\Bigg(#2\Bigg)}%
}[\PackageError{paren}{Undefined option to paren: #1}{}]%
}

\newcommand{\set}[2][a]{
\IfEqCase{#1}{%
{a}{\left\{#2\right\}}%
{0}{\{#2\}}%
{1}{\big\{#2\big\}}%
{2}{\Big\{#2\Big\}}%
{3}{\bigg\{#2\bigg\}}%
{4}{\Bigg\{#2\Bigg\}}%
}[\PackageError{set}{Undefined option to set: #1}{}]%
}

\newcommand{\brac}[2][a]{%
\IfEqCase{#1}{%
{a}{\left[#2\right]}%
{0}{[#2]}%
{1}{\big[#2\big]}%
{2}{\Big[#2\Big]}%
{3}{\bigg[#2\bigg]}%
{4}{\Bigg[#2\Bigg]}%
}[\PackageError{brac}{Undefined option to brac: #1}{}]%
}

\newcommand{\cH}{\mathcal{H}}

\newtheorem{thm}{Theorem}
\newtheorem{prp}[thm]{Proposition}
\newtheorem{lem}[thm]{Lemma}
\newtheorem{cor}[thm]{Corollary}
\newtheorem{dfn}[thm]{Definition}
\newtheorem{ass}{Assumption}

\theoremstyle{remark}
\newtheorem{ex}{Example}
\newtheorem{rem}{Remark}

\usepackage{MnSymbol}

\def\beq{\begin{equation}} 
\def\eeq{\end{equation}}
\def\beqn{\begin{eqnarray*}}
\def\eeqn{\end{eqnarray*}}
\def\Bal{\begin{align}}
\def\Eal{\end{align}}
\def\Bitem{\begin{itemize}\setlength{\itemsep}{.2in}}
\def\bitem{\begin{itemize}\setlength{\itemsep}{.05in}}
\def\eitem{\end{itemize}}
\def\blatin{\begin{enumerate}\setlength{\itemsep}{.05in}\renewcommand{\labelenumi}{\roman{enumi}.}}
\def\elatin{\end{enumerate}}
\def\Benum{\begin{enumerate}\setlength{\itemsep}{.2in}}
\def\benum{\begin{enumerate}\setlength{\itemsep}{.05in}}
\def\eenum{\end{enumerate}}
\def\bmult{\begin{multline*}}
\def\emult{\end{multline*}}
\def\bcenter{\begin{center}}
\def\ecenter{\end{center}}
\def\bframe{\begin{frame}}
\def\eframe{\end{frame}}

\def\cA{\mathcal{A}}

\def\cH{\mathcal{H}}

\def\cL{\mathcal{L}}

\def\cN{\mathcal{N}}
\def\cO{\mathcal{O}}
\def\cP{\mathcal{P}}
\def\cQ{\mathcal{Q}}

\def\cX{\mathcal{X}}

\def\cZ{\mathcal{Z}}

\def\bbE{\mathbb{E}}

\def\bbO{\mathbb{O}}
\def\bbP{\mathbb{P}}

\def\ind{\mathbbm{1}}

\newcommand{\expect}[1]{\mathbb{E}\left[#1\right]}

\newcommand\indep{\protect\mathpalette{\protect\independenT}{\perp}}
\def\independenT#1#2{\mathrel{\rlap{$#1#2$}\mkern2mu{#1#2}}}

\def\implies{\ \Rightarrow \ }

\newcommand{\ve}{\varepsilon}
\newcommand{\vp}{\varphi}
\newcommand*\diff{\mathop{}\!\mathrm{d}}
\renewcommand{\leq}{\leqslant}
\renewcommand{\geq}{\geqslant}

\makeatletter
\newcommand{\para}{\mathrel{\mathpalette\new@parallel\relax}}
\newcommand{\new@parallel}[2]{%
  \begingroup
  \sbox\z@{$#1T$}
  \resizebox{!}{\ht\z@}{\raisebox{\depth}{$\m@th#1/\mkern-5mu/$}}%
  \endgroup
}
\makeatother

\def\independenT#1#2{\mathrel{\rlap{$#1#2$}\mkern2mu{#1#2}}}

\usepackage{cancel}

\newcommand{\transpose}{\mathsf{T}}

\title{Formalization and consistency of clone-censor-weight\\[.5em]}

\author[1]{Iqraa Meah$^{\diamond}$}
\author[2]{Charlotte Voinot$^{\diamond}$}
\author[3]{Stefan Michiels}
\author[3]{Noémie Simon-Tillaux}
\author[4]{Claire Jouffroy-Zeller}
\author[5]{Inès Vaz-Luis}
\author[6]{Clément Berenfeld$^{\ddagger}$}
\author[7]{François Petit$^{\ddagger}$}

\affil[1]{UMR MIA Paris-Saclay, AgroParisTech, INRAE, U. Paris-Saclay}
\affil[2]{Laboratoire Servier}
\affil[3]{Office of Biostatistics and
Epidemiology, Gustave Roussy,
Oncostat U1018, Inserm,
U. Paris-Saclay}
\affil[4]{R\&D Unicancer}
\affil[5]{U. Paris-Saclay, UVSQ, Inserm, CESP U1018, Cancer Survivorship Group, Gustave Roussy}
\affil[6]{PreMeDICaL, Inria-Inserm, U. Montpellier, France}
\affil[7]{U. Paris Cité and U. Sorbonne Paris Nord, INSERM, CRESS}

\date{}

\begin{document}

\maketitle

\begingroup
\renewcommand{\thefootnote}{}
\footnotetext{$^{\diamond}$ These authors share first authorship. $^\ddagger$ These authors contributed equally to the supervision of this work and share last authorship.}
\addtocounter{footnote}{-1}
\endgroup



\vspace{-1.5em}
\begin{abstract}
    Target trial emulation has become a standard framework for causal inference from observational data. Within this paradigm, the 
    clone-censor-weight (CCW) methodology provides a practical way to deal with complex treatment and adherence definitions when treatment regimes are not distinguishable at baseline. Despite its increasing use, the statistical interpretation of CCW remains limited. In this work, we formalize the CCW methodology from a statistical viewpoint based on admissible sets and stochastic interventions. We characterize the causal estimands targeted by CCW, establish consistency of CCW estimators under standard identification assumptions, and provide theoretical guarantees for bootstrap-based inference methods.
\end{abstract}

\section{Introduction} \label{sec:intro}

Target trial emulation, introduced by \cite{hernanRobins2016}, provides a formal framework for conducting comparative studies from observational data. The approach consists in explicitly defining the protocol of an ideal randomized trial (through the definition of eligibility criteria, treatment strategies, treatment assignment, estimands and contrasts) and then emulating this protocol using observational data. This methodology has become central in epidemiology and clinical research, as it clarifies what is being estimated and help to reduce biases induced by fallacious comparisons. \medskip

A central feature of randomized trials is the distinction between intention-to-treat (ITT) and per-protocol (PP) analyses. ITT analyses compare groups as assigned at baseline, regardless of subsequent adherence. In contrast, PP analyses aim to estimate the effect of sustained adherence to a treatment strategy, typically by excluding or censoring individuals who deviate from the protocol \citep{Hollis670}. Since the notion of "treatment assignment" is most often not explicitly recorded in observational data 
ITT analyses are generally discarded, and the focus is made on PP estimands, which require defining adherence in a way clinically meaningful while remaining practically implementable from the observed trajectories. \medskip

A naive approach to comparing treatment strategies from observational data consists in partitioning individuals based on post-baseline information, such as whether and when they initiate treatment. This approach is generally invalid and can induce bias, notably the immortal time bias in time-to-event settings \citep{suissa2008immortal,yadav2021immortal,maringe2020reflection}. 
Immortal time bias arises when individuals assigned to a group must survive long enough to satisfy its defining criteria. For example, when comparing treatment durations (e.g., ``treat six months'' vs.\ ``treat three months''), grouping patients by their observed treatment duration requires those in the ``treat six months'' group to have survived at least six months, artificially inflating the group's survival. This can severely bias effect estimates and compromise the validity of the analysis.\medskip


The clone-censor-weight (CCW) methodology is an approach that overcomes these issues. Originally proposed in \cite{cain2010start}, CCW arranges identical pseudo-populations (through cloning) in which individuals are sequentially removed (reffered as \emph{artificial} censoring) from the population when their observed trajectory deviates from the treatment protocol defining the population.  Individuals are then reweighted to account for the selection bias induced by the artificial censoring. This procedure ensures comparability across the different arms and has been increasingly used in clinical epidemiology studies, see e.g. \citep{dickerman2019avoidable}. \medskip



However, defining adherence in the context of observational data can be challenging. First, treatment strategies may be broadly defined, such that several types of trajectories can qualify as adherent.
Second, even for more prescriptive strategies, practical considerations such as positivity constraints can complicate the definition of adherence. A common example is the use of grace periods, which allow treatment initiation within a prespecified time window. Such periods may be clinically motivated (for instance, in surgical settings where immediate intervention is unrealistic, as for instance in \cite{kirkegaard2024effect}) but they may also be needed to ensure sufficient treatment positivity in the data, for instance when studying treatment duration as in \cite{amiot2024stop}. \medskip 

Despite its increasing use, a formal statistical study of CCW that accommodates general definitions of adherence and provides theoretical guarantees remains, to our knowledge, to be developed. In many applied studies, the focus is placed on implementation without a precise characterization of the causal estimand being targeted. In particular, different practical choices in the design of the analysis, such as the length of grace periods are frequently treated as technical details, even though they define distinct interventions and therefore distinct causal estimands \citep{wanis2024grace}. Such detail appears also to be important in other context e.g. when studying discontinuation strategies. Furthermore, these estimands may themselves depend on the data distribution in intricate ways. As a result, two identical analyses based on different observational data may not be computing the same estimand, despite appearing to address exactly the same clinical question. \medskip

Although the CCW approach has been illustrated in several applied examples and is closely related to other causal inference frameworks, including marginal structural models, and dynamic treatment regimes \citep{robins2000marginal,Robins1986,keil2014gformula,naimi2017gmethods,tsiatis2020}, a general formalization of the methodology has been lacking. Methodological work has only recently begun to investigate the question of targeted estimands, notably for treatment initiation windows \citep{webster2024complex} or for sustained treatments with grace periods \citep{wanis2024grace}, where the estimand may depend on the empirical data distribution in the study population.\medskip

This work proposes a statistical formalization of the CCW methodology for longitudinal time-to-event data and establishes the theoretical properties of CCW estimators. We introduce a unified framework encompassing a broad class of treatment strategies and adherence definitions, characterize the target estimand, prove consistency under standard identification assumptions, and justify bootstrap-based inference. The framework also enables sensitivity analyses to key design choices, such as the length of grace periods.\medskip



The rest of the paper is organized as follows. Section~\ref{sec:setting} introduces the formal framework of the paper. Section~\ref{sec:consistency} presents the clone-censor-weight methodology and establishes the consistency of the CCW estimator, while Section~\ref{sec:bootstrap} discusses the validity of bootstrap-based inference in this context.  Section~\ref{sec:simulation} illustrates the theoretical results through a serie of simulation study, highlighting the impact of the choice of adherence definitions (e.g., grace periods) on the final estimators and estimands.

\section{Setting} \label{sec:setting}
\subsection{General notations} \label{subsec:notations}

We start by introducing the notations of the paper. 
Although the formalism we present is valid for any kind of outcome, we focus on time-to-event outcomes.
For the remaining, unless stated otherwise, the subscrit $t \leq \tau$ denotes a discret time index with $\tau \in \mathbb{N}$.
For a time-varying variable $Z_t$, we denote with an overline the variable history from time $0$ to time $t$: $\overline Z_t := (Z_0,\dots,Z_t)$. 
Likewise, we denote by an underline the future values of the variable: $\underline Z_t := (Z_t, Z_{t+1}, \dots)$. We do the same for Cartesian products of sets: if $\cZ_t$ is a sequence of sets, then $\overline \cZ_t := \cZ_0 \times \dots \times \cZ_t$. \medskip

For each patient, we observe an ordered sequence 
\begin{equation} \label{eq:observations}
	X_0, A_0, \qquad C_1, Y_1, X_1,A_1, \qquad \dots \qquad C_t, Y_t, X_{t},A_{t}, \qquad \dots \qquad
	C_\tau, Y_\tau, X_{\tau};
\end{equation} 
where $X_t \in \cX_t$ are measured covariates at time $t$ (often the same covariates are measured over time and $\cX_t=\cX$ for all $t$);
 	 $A_t \in \mathcal{A}_t$ is the treatment allocation indicator at time $t$ (often $\cA_t=\{0,1\})$; $C_t \in \{0,1\}$ is the natural censoring status ($C_t=1$ if the individual has been lost to follow-up between time $t-1$ and $t$); and $Y_t \in \{0,1\}$ is the event  indicator ($Y_t=1$ if the individual has experienced the event between time $t-1$ and $t$). We let
     $$O_t = (X_0,A_0,\dots,C_t,Y_t,X_t) \in \cO_t,$$ 
     be the total observations up to time $t$. \medskip

We assume that the follow-up of the patient stops when $Y_t = 1$ , $C_t = 1$, or $t = \tau$, whichever happens first. After $C_t=1$, the remaining variables exists but are not observed, and $C_s$ stays to $1$ for $s > t$. After $Y_t=1$, the values of $Y_s$ are set to $1$ and the ones of $C_s$, $X_s$ and $A_s$ are set to a dummy value $\aleph$. Finally, we let $H_t = (\overline X_t, \overline A_{t-1}) \in \cH_t$, where $\cH_t := \overline \cX_t \times \overline \cA_{t-1}$. The variable $H_t$ is called the history and corresponds to the information available at time $t$ that should help to take the treatment decision $A_t$. Note that with this convention $Y_t$ is $X_t$-measurable since $Y_t = \ind\{X_t = \aleph\}$; for this reason, we exclude $Y_t$ from the history variable $H_t$, when it is usually included in the literature. We summarize all these variables in Table~\ref{tab:variables}.

\begin{table}[ht]
\centering
\begin{tabular}{llll}
\hline
Time & Quantity & Domain & Interpretation \\
\hline

interval $(t-1,t]$ & $C_{t}$ & $\{0,1\}$ & censoring occurs in $(t-1,t]$ \\

interval $(t-1,t]$ & $Y_{t}$ & $\{0,1\}$ & event occurs in $(t-1,t]$ \\

at time $t$ & $X_{t}$ & $\cX_t$ & updated covariates at time $t$\\

at time $t$ & $O_t$ & $\cO_t$ & all observations available at time $t$ \\

at time $t$ & $H_t$ & $\cH_t$ & data guiding treatment decision at time $t$ \\

right after time $t$ & $A_t$ & $\cA_t$ & treatment decision at time $t$ \\




\hline
\end{tabular}
\caption{Summary of the different observed variables.} \label{tab:variables}
\end{table}

\subsection{Admissible trajectories and stochastic interventions}


In many clinical epidemiology reseach context, formulation of treatment rules is broad enougth to allow for flexibility, e.g. in treatment initiation timing or temporary deviations, resulting in the absence of a single prototypical trajectory.
For instance, in studies including grace periods, patients who initiate treatment at different times within the allowed period are all regarded as adhering to the same strategy, even though their observed trajectories differ. This phenomenon extends beyond initiation rules, with e.g. persistence or discontinuation strategies that allow gaps in the treatment trajectory, see e.g. \cite{Dumas_2025}. \medskip

From a statistical standpoint, restricting attention to a single trajectory would lead to severe positivity issues, since only a negligible fraction of the observed data would exactly match that trajectory. Defining a (broader) set of admissible trajectories alleviates these positivity concerns by enlarging the support of the intervention, while preserving the clinical meaning of the strategy.


\begin{dfn}[Admissible trajectories] 
    Let $\Omega = \{\Omega_t\}_t$ be a sequence of sets of trajectories $\Omega_t \subset  \cH_t \times \cA_t$. We say that $\Omega$ is a sequence of \emph{admissible trajectories} if, for all $s \leq t$ the traces of trajectories of $\Omega_t$ on $\cH_s \times \cA_s$ is contained in $\Omega_s$. 
\end{dfn}

\begin{ex}[Surgery with grace period] \label{ex:surgery}

\noindent\cite{maringe2020reflection} compare two treatments strategies for lung cancer: surgery within six months after diagnosis vs. no surgery within the first six months. Time zero $t=0$ denotes the time of diagnosis, while subsequent time indices refer to the $t^{th}$ monthly visit. We assume that all patients undergo surgery at most once. Since the intervention is static, i.e does not depend on the evolution of the covariate process $X_t$, we can remove the covariates from the trajectory to only keep the sequence of treatment allocation  $\overline a_t \subset \{0,1\}^{[t]}$, with $a_t = 1$ denoting that surgery happened at time $t$. Letting $|\overline a_t|$ denotes the sum of all entries of $\overline a_t$, the set of admissible trajectories in this case is given by
\begin{align*}
	\Omega_t := 
	\begin{cases}
		\{\overline{a_t} \mid |\overline{a_t}| \leq 1 \} & \text{if } t \leq 5\\
		\{\overline{a_t} \mid |\overline{a_t}| \leq 1, |\overline{a_6}| = 1 \} & \text{if } t \geq 6. 
	\end{cases}
\end{align*}
\end{ex}

A common implicit assumption is that all trajectories within a given set are equivalent in terms of their causal effect, so that the strategy can be represented solely through this support. This is often unrealistic in practice: trajectories that differ in timing or duration may lead to different effect on the outcome, even if they are all clinically acceptable versions of the same rule. Since different trajectories may lead to different outcome effects, it is necessary to complement the admissibility constraints with a mechanism that assigns probabilities to admissible actions at each time. This leads to the notion of stochastic interventions, which  specify how trajectories are generated within the admissible sets.

\begin{dfn}[Stochastic intervention]
A stochastic intervention $q$ is a family of kernel
$$ q =\{q_t:  \cH_{t} \to \cP(\mathcal{A}_{t}),~t \in [\tau-1]\},$$ 
that, given the available information $H_{t}$ at time $t$, provides a treatment distribution $q_t(\cdot \mid  H_{t})$ on $\cA_t$ for the patient at time $t$.
\end{dfn}

\begin{rem} A deterministic intervention $q^*$ can be seen as a stochastic intervention through the degenerate kernel
$$
q_t(\cdot \mid H_{t}) = \delta_{q_t^*(H_{t})},
$$
where $\delta$ denotes the Dirac distribution. We thus make no difference between the two in the following.
\end{rem}


For a given set of admissible trajectories, there is generally a large class of stochastic interventions that induce exactly the same support, differing only in how probability mass is distributed across admissible actions at each time point. In other words, admissibility alone does not identify a unique intervention, but rather a family $\cQ(\Omega)$ of data-generating mechanisms compatible with the same clinical rule. 

\begin{dfn}[Compatible stochastic intervention]
    Let $\Omega$ be a sequence of admissible trajectories. We define the set of compatible stochastic intervention $\cQ(\Omega)$ as the set of all intervention $q$ such that $q_s(A_s\mid H_{s}) > 0$ if, and only if $(H_s,A_s) \in \Omega_s$. 
\end{dfn}

Among all the compatible intervention, one plays a distinguished role: the one induced by the observational data-generating process. This \emph{natural} stochastic intervention preserves the empirical treatment assignment probabilities while enforcing the admissibility constraint, and therefore serves as a canonical choice in the analysis.

\begin{dfn}[Natural intervention] \label{def:qnat}
    Among all possible kernels of $\cQ(\Omega)$, the one induced by the distribution of the data is called the \emph{natural} intervention associated with $\Omega$, and is denoted by $q^{\rm nat}_\Omega$. Formally, it is defined as 
    $$
    q_{\Omega,t}^{\rm nat}(a \mid H_{t}) \propto \begin{cases}
     0 \quad &\text{if} \quad a \notin \cA_{\Omega,t}(H_{t}) \\
     \bbP\left(A_t = a \mid H_{t}, Y_t=C_t=0\right) \quad &\text{otherwise},
    \end{cases}
    $$
    where $\cA_{\Omega,t}(H_{t})$ is the set of $a \in \cA_t$ such that $(H_{t},a) \in \Omega_t$.
\end{dfn}



\begin{ex}[Example \ref{ex:surgery}, continued] \label{ex:surgery-cont}
The natural stochastic intervention $q^{\rm nat}$ associated with the admissible trajectories of Exemple~\ref{ex:surgery} is given by
\begin{align*}
	q_{\Omega,t}^{\rm nat}(a \mid H_t)  = 
	\begin{cases}
		\ind\{a=0\} \quad &\text{if} \quad t \geq 7 \quad \text{or} \quad~~ |\overline A_{t-1}| =1,\\
		 \bbP(A_t=a \mid \overline X_t, \overline A_{t-1} = 0, Y_t=C_t=0)  &\text{if} \quad  t < 6 \quad \text{and} \quad \overline A_{t-1}= 0,\\
      \ind\{a=1\}  &\text{if} \quad  t = 6 \quad \text{and} \quad \overline{A}_5 = 0.
	\end{cases}
\end{align*}
Note that, even though the intervention does not depend on the covariates $\overline X_t$, the natural intervention might.
\end{ex}
We refer to Section~\ref{sec:examples} for other examples including discontinuation and persistence strategies.

\subsection{Causal inference framework} \label{subsec:causalinf-framework}

To draw a causal comparison of different interventions, one needs to define what the outcome would have been under each of them, possibly contrary to fact. To do so, we  embed the previous formalism into the potential outcome framework, where each stochastic intervention is associated with a collection of counterfactual variables. Within this framework, causal effects can be expressed as contrasts between intervention-specific outcomes, and their identification from observational data relies on a set of standard assumptions, which we introduce in this section.

\paragraph{Potential outcomes.} For all $t \in [\tau]$ and all $\overline a_t \in \overline{\cA}_t$, we posit the existence of the counterfactual variables $Y_t(\overline a_{t-1})$, $X_t(\overline a_{t-1})$ and $\overline C_t(\overline a_{t-1})$, which correspond to the values of $Y_t$, $X_t$ or $C_t$ respectively if, possibly contrary to the fact, the treatment allocated between time $0$ and $t-1$ had been $\overline a_{t-1}$. By convention, if $\overline a \in \overline A_\tau$ is a full-time treatment trajectory, we will write $Y_t(\overline a)$ for $Y_t(\overline a_{t-1})$, and likewise for $X_t(\overline a)$ and $C_t(\overline a)$. This will allow us to consider, for instance, the counterfactual future process $\underline Y_t(\overline a)$, etc.

Given a stochastic intervention $q$, we define the potential outcomes $Y_t(q)$, $X_t(q)$ and $C_t(q)$ as the potential observations that would occur at time $t$, if possibly contrary to the fact the stochastic intervention $q$ had been followed from time $0$ to $t-1$.

\paragraph{Causal target and contrast.} Several causal estimands can be considered depending on the outcome of interest and the chosen scale. In the potential outcome framework, each intervention $q$ induces a counterfactual outcome process from which one can define, for instance, the  survival function $S^q(t) := \bbP(Y_t(q)=0)$ or functionals of the latter such as the restricted mean survival time RMST):
$$
\rho^q(\tau) := \int_0^\tau S^q(t) \diff t.
$$
Given two interventions $q_0$ and $q_1$, causal contrasts are obtained by comparing the two associated causal estimands either through differences (survival difference, RMST difference) or ratios (risk ratio, hazard ratio) depending on the scientific question. In this paper, we will focus either on the difference in survival curve or on the RMST:
\begin{equation}
\Delta S(t) := S^{q_1}(t) - S^{q_0}(t)
\qquad \text{and} \qquad
\Delta \rho(\tau) := \rho^{q_1}(\tau) - \rho^{q_0}(\tau).
\label{eq:contrats-definitions}
\end{equation}

\paragraph{Identification assumptions} Although causal estimands are well-defined in terms of counterfactual outcomes, they are not directly observable from the data, since each individual is only observed under the treatment actually received. Working through this setting of "missing data" requires a set of identification assumptions that link the distribution of counterfactual outcomes under the interventions of interest to the observed data-generating process. 
Without them, the causal contrasts defined above cannot be expressed or estimated from the observed data.

\begin{ass}[Consistency] \label{ass:consistency} If $\overline A_{t-1} = \overline a_{t-1}$, then
   \[
   Y_t=Y_t(\overline{a}_{t-1}),  \quad X_t=X_t(\overline{a}_{t-1}), \quad \text{and}\quad C_t=C_t(\overline{a}_{t-1}).
   \]
\end{ass}
\begin{ass}[Sequential ignorability] \label{ass:ignorability} For all trajectories of treatment $\overline a$,
    \[
    \underline{Y}_{t+1}(\overline a), \underline{X}_{t+1}(\overline a), \underline{C}_{t+1}(\overline a) \indep  A_t \mid \overline A_{t-1}=\overline a_{t-1}, \overline X_t, Y_{t}=C_t=0
    \] 
\end{ass}

\begin{ass}[Conditionally independent censoring] \label{ass:censoring} For all treatment trajectory $\overline a$, 
    $$
    \underline{X}_{t+1}(\overline a),\underline{Y}_{t+1}(\overline a) \indep  C_{t+1}(\overline a) \mid \bar A_{t}=\bar a_t,\overline X_t, Y_t=C_t=0.
    $$
\end{ass}
\begin{ass}[Treatment positivity] \label{ass:positivity_treatment} There exists $\ve > 0$ such that
     \[
    q_t(a\mid H_{t}) > 0 \implies \bbP(A_t = a \mid H_{t},Y_t=C_t=0) > \ve \quad \text{almost-surely}. 
    \] 
\end{ass}

\begin{ass}[Censoring positivity]\label{ass:positivity_censoring}
There exists $\ve > 0$ such that
    \[
    q_t(A_t \mid H_{t})>0 \implies \bbP(C_{t+1}=0 \mid A_t, H_t, Y_t=C_t=0) > \ve \quad \text{almost-surely.}
    \]
\end{ass}

Consistency (Assumption~\ref{ass:consistency}) states that the observed outcome coincides with the counterfactual outcome under the realized treatment history. Sequential ignorability (Assumption~\ref{ass:ignorability}) requires that, conditional on the observed history, treatment assignment at each time is independent of the future counterfactual process, ruling out unmeasured time-varying confounding. Conditional independent censoring (Assumption~\ref{ass:censoring}) plays an analogous role for loss to follow-up, ensuring that censoring does not depend on the future counterfactual process given the observed past. The last two assumptions are dependent on the intervention being evaluated. The positivity assumptions for treatment (Assumption~\ref{ass:positivity_treatment}) and censoring (Assumption~\ref{ass:positivity_censoring}) guarantee that the intervention of interest is supported by the data: all actions under the considered intervention must occur with nonzero probability in the observed process. 

\section{Consistency of the clone-censor-weight method} \label{sec:consistency}

\subsection{The clone-censor-weight methodology}

In observational settings, partitioning the patient data into groups based on post-baseline information can induce bias, for instance when defining groups according to future treatment decisions (e.g., "treated within 6 months" versus "not treated"). In time-to-event settings, this typically leads to the so-called \emph{immortal time bias} \citep{suissa2008immortal}: individuals classified in a given group at a specific time must, by construction, survive long enough to satisfy the defining condition, which artificially inflates survival in that group. \medskip

The clone-censor-weight (CCW) approach addresses this issue by circumventing the partition paradigm of the observed data. First, each individual is cloned into as many copies as there are interventions being compared. Second, each clone is \emph{artificially} censored (as opposed to \emph{naturally} through $C_t$) at the first time its observed trajectory $(H_t,A_t)$ becomes incompatible with the admissible set $\Omega_t$ defining the assigned intervention. 
Finally, this artificial censoring induces selection bias that is corrected through inverse probability weighting (IPW). The latter step, reweights the remaining observations to recover, under the identification assumptions, the distribution of outcomes that would have been observed had all individuals followed the intervention. Estimation then proceeds by applying standard survival estimators within each cloned population. \medskip 

The CCW approach can be viewed as a 
way to overcome immortal time bias issues 
\citep{hernanRobins2016} in longitudinal observational settings. In a target trial, individuals are randomized at baseline to well-defined treatment strategies and followed untill the completion of the study. CCW reconstructs this ideal experiment, with each weighted pseudo-population corresponding to the PP population that would have been observed under the target trial. We summarize in Table~\ref{tab:trialprotocol} how CCW allows emulating a target trial. 

\begin{table}[ht]
\centering
\begin{tabular}{|p{5cm}|p{9cm}|}
\hline
\textbf{Target trial's protocol} &  \textbf{CCW methodology} \\
\hline
Eligibility criteria & Selection at baseline uppon $X_0$ \\
\hline 
Treatment strategies & 
Stochastic interventions $q_0$ and $q_1$ \\
\hline
Treatment assignment &   Randomization is emulated
via cloning\\
\hline 
Start and end of follow-up & 
Time $t=0$ and $t=\tau$ \\
\hline 
Outcome  & Survival time
\\
\hline 
Causal contrast  & Survival probability difference $\Delta S(t)$ or RMST $\Delta \rho(\tau)$, see \eqref{eq:contrats-definitions}
 \\
\hline 
Statistical analysis & Per protocol, enforced by artificial censoring \\
\hline 
\end{tabular}
\caption{Protocol elements of the target trial and its counterparts in the CCW framework}
\label{tab:trialprotocol}
\end{table}


\subsection{Definition of the clone-censor-weight estimators} 

Formally, given an intervention $q$ and its associated admissible trajectories $\Omega$, we define
$$
\Gamma_{\Omega,t} := \ind\{(H_t,A_t) \in \Omega_t\} = \ind\{A_t \in \cA_{\Omega,t}(H_t)\} ,
$$
which is the artificial censoring variable associated with the intervention $q$ at time $t$. We let 
$$
p^C_t(A_{t-1}, H_{t-1})
:= \bbP(C_t = 0 \mid A_{t-1}, H_{t-1}, C_{t-1}=Y_{t-1}=0).
$$ 
be the natural censoring probability, and 
$$
p^A_t(a \mid H_{t}) := \bbP(A_t = a \mid  H_{t}, C_{t}=Y_{t}=0)
$$
be the probability of observed treatment. We define the \emph{at-risk} indicator $R_{t}$ to be
$$
R_t := \prod_{s=1}^t \ind\{C_s = 0, Y_{s-1} = 0\}. 
$$
The variable $R_t$ select the individual that are still alive at time $t$ and that haven't been censored. We could add the artificial censoring variables $\Gamma_{\Omega,t}$ in $R_t$, but we can also omit it as it will appear implicitly in the definition of the weights below.
Likewise, the event indicator $E_t$ is
$$
E_t := \ind\{Y_{t-1}=0, Y_t =1\}.
$$
Finally, we define the weight at time $t$ as:
$$
W_{t}^{q}
:= \prod_{s=0}^{t}
\frac{q_s(A_{s} \mid H_{s})}{p^A_s(A_{s} \mid H_{s})}
\times
\frac{1}{p^C_{s+1}(A_{s},H_s)}.
$$

\begin{rem}
When censoring is independent, there is no need to model the censoring mechanism as the weights will cancel themselves in the Kaplan-Meier estimator as they do not depend of history. By contrast, the natural censoring weights must be modelled and estimated when censoring is informative. This may arise, for example, when the prognosis of newly enrolled patients changes over calendar time, inducing a shift in the case mix towards patients with better or worse prognoses. Such changes may occur, for instance, as a consequence of improvements in the standard of care.
\end{rem}

All quantities related to an individual $i$ will be denoted with a subscript $i$. Note that by design, $W_{i,t}^q$ assigns a weight zero to the individual $i$ if the latter has deviated from the intervention $q$ at or before time $t$, that is if $\Gamma_{\Omega,i,t} = 0$. The weights thus act as an \emph{artificial} censoring filter and enforce the adherence to the intervention $q$. Assuming that there is $n$ individuals in the observed data, we define 
\begin{equation} \label{eq:nd}
N_t^q := \sum_{i=1}^n W_{i,t-1}^q R_{i,t} \quad \text{and} \quad D_t^q := \sum_{i=1}^n W_{i,t-1}^q R_{i,t} E_{i,t}. 
\end{equation}
The CCW estimator of $S^q(t)$ is then defined as 
\begin{equation}\label{eq:ccw}
\hat S^q_{\rm CCW}(t) := \begin{cases} \displaystyle \prod_{s \leq t}\left(1 - \frac{D_s^q}{N_s^q}\right) \quad &\text{if} \quad N_t^q \neq 0,\\
0 \quad &\text{otherwise.}
\end{cases}
\end{equation}
This corresponds to a usual Kaplan-Meier type estimator under a IPCW reweighting scheme, see \cite{RobinsFinkel200}.

\subsection{Consistency results}

We turn to the theoretical properties of the CCW estimators introduced above. In particular, we establish their consistency for the causal estimands $S^q(t)$ defined in Section~\ref{sec:setting} and  under the identification assumptions stated in the same section. 

\begin{thm} \label{thm:consistency} Under Assumptions \ref{ass:consistency}-\ref{ass:positivity_censoring}, 
$\hat S^{q}_{\rm CCW}(t)$ is a consistent estimator of
$
S^{q}(t).
$
\end{thm}

The proof of this result can be found in Appendix~\ref{app:consistency}. Usual applications of the CCW methodology \citep{cain2010start, Cain2016} use the (unstabilized) weights 
$$
W_{i,t}^{\Omega}
:= \prod_{s=0}^{t}
\frac{\Gamma_{\Omega,i,s}}{p^\Omega_s(H_{i,s})}
\times
\frac{1}{p^C_{s+1}(A_{i,s},H_{i,s})},
$$
where 
$$
p^\Omega_s(H_{s}) := \bbP((H_s,A_s) \in \Omega_s \mid H_{s}, Y_s=C_s=0). 
$$
We call this version the \emph{natural} CCW estimator. It only depends on the admissible trajectories $\Omega$ and does not rely on the specification of a compatible stochastic intervention $q$. We thus denote this estimator by $\hat S^{\Omega}_{\rm CCW}(t)$.   This next result clarifies the precise intervention targeted by the natural CCW procedure.

\begin{cor}\label{cor:natural_weights} Under Assumptions \ref{ass:consistency}-\ref{ass:positivity_censoring}, the \emph{natural} CCW estimator $\hat S^{\Omega}_{\rm CCW}(t)$ is a consistent estimator of $S^{q^{\rm nat}_\Omega}(t)$ where $q^{\rm nat}_\Omega$ is the natural intervention associated with $\Omega$. 
\end{cor}

This follows for the fact that $W_{i,t}^\Omega = W_{i,t}^{q^{\rm nat}_\Omega}$ and the application of Theorem~\ref{thm:consistency}.

\begin{rem}
The natural weights corresponds to the unstabilized weights used in marginal structural models (see \cite{robins2000marginal}). However, weight stabilization does not operate in the same way as in marginal structural models. In the Kaplan–Meier-type estimator considered here, stabilizing the censoring weights generally changes the causal estimand, leading to the estimation of a different survival quantity, corresponding to a $S^q$  where the stochastic intervention $q$ is determined by the stabilization scheme.
\end{rem}

\section{Bootstrap-based inference in the CCW framework} \label{sec:bootstrap}

In medical applications, quantifying the uncertainty around CCW estimates is as crucial as point estimation. Closed-form variance expressions exist in principle, but they are often difficult to derive and to implement in practice, especially in complex longitudinal settings. As a result, practitioners frequently rely on bootstrap methods as a convenient alternative. The bootstrap \citep{efron1994introduction} is a resampling procedure in which one repeatedly draws samples, with replacement, from the original dataset at the individual level, recomputes the full estimator on each resampled dataset, and uses the empirical distribution of these replicated estimates to approximate the sampling distribution of the estimator. Namely, from $n$ observed individual data $\bbO_n := (O_{1},\dots,O_{n})$, we re-sample $n$-times with replacement to get a distribution of the form
\begin{equation}\label{eq:pstar}
\bbP^*_{n} = \frac{1}{n}\sum_{i=1}^n \mathfrak{m}_{i} \delta_{O_{i}} \quad \text{where} \quad \mathfrak{m} := ({\mathfrak{m}}_{1},\dots,\mathfrak{m}_{n}) \sim {\rm Mult}(n,1/n,\dots,1/n),
\end{equation}
and where $\mathfrak{m}$ is sampled independently of $\bbO_n$. We call $\bbP^*_n$ a \emph{bootstrap} distribution. \medskip

We let $\Omega_0$ and $\Omega_1$ be two sets of admissible trajectories. Our goal is to estimate the variance of the CCW estimator of either of the contrasts:
$$
\Delta S(t) = S^{q^{\rm nat}_{\Omega_1}}(t) - S^{q^{\rm nat}_{\Omega_0}}(t) \quad \text{or} \quad \Delta \rho(\tau) =  \rho^{q^{\rm nat}_{\Omega_1}}(\tau) - \rho^{q^{\rm nat}_{\Omega_0}}(\tau).
$$
We denote by $\Delta$ either one of these estimands and $\hat \Delta_n$ the corresponding CCW estimator. The computation of $\hat\Delta_n$ require the knowledge of $W_{t-1}^{\Omega_1}$ and $W_{t-1}^{\Omega_0}$, and in turn the one of $p_t^C$, $p_t^{\Omega_1}$ and $p_t^{\Omega_0}$ for $t\in\{1,\dots,\tau\}$. We parametrize these functions by some finite dimensional \emph{nuisance} parameters $\eta$, and we assume that we estimate it by some Z-estimator $\hat\eta_n$ as described in detail in Appendix~\ref{app:bootstrap}. This covers various usual estimation procedure, such as for instance, pooled logistic regressions \citep{d1990relation}.
We let $\hat\Delta_n^*$ be the CCW bootstrap estimator built from $\bbP^*_n$ as in \eqref{eq:pstar} and using Eqs~\eqref{eq:nd}-\eqref{eq:ccw} with $\hat\eta_n^*$ being the Z-estimator computed from $\bbP^*_n$.  

\begin{thm}
\label{thm:bootstrap}
Under Assumptions~\ref{ass:consistency}-\ref{ass:positivity_censoring} and technical conditions given in Appendix~\ref{app:bootstrap_assumption} (\ref{ass:Zidentifiability}--\ref{ass:non_sing} and \ref{ass:bootstrap_sampling}--\ref{ass:bootstrap_precision}), it holds 
\begin{enumerate}
    \item The CCW estimator $\hat \Delta_n$ is asymptotically normal with 
    $$
    \sqrt{n}(\hat \Delta_n- \Delta) \xrightarrow[n\to\infty]{\cL} \cN(0,\sigma_\Delta^2). 
    $$
\item The bootstrap CCW estimator $\hat \Delta^*_n$ is such that, for all bounded continuous function $\vp$,
$$
\bbE_[\vp(\sqrt{n}(\hat \Delta^*_n-\hat \Delta_n)) \mid \bbO_{n}] \xrightarrow[n\to\infty]{\bbP} \bbE[\vp(Z)] \quad \text{where} \quad Z \sim \cN(0,\sigma_\Delta^2),
$$
and where $\bbO_{n}=( O_{1},\dots, O_{n})$.
\end{enumerate}
\end{thm}

We prove these results in Appendix~\ref{app:bootstrap}, and illustrate them in Section~\ref{sec:simulation}. We show in the next result that the technical assumptions of Appendix~\ref{app:bootstrap_assumption} are in particular satisfied for two widely-used nuisance parameter estimation methods, namely the pooled logistic regression and the log-log regression. 
\begin{cor}
    Assumptions \ref{ass:Zidentifiability}--\ref{ass:non_sing} holds under the Assumptions~\ref{ass:glm1}--\ref{ass:glm6} given in Appendix~\ref{app:pooled-glm}. Hence, under the bootstrap assumptions \ref{ass:bootstrap_sampling}--\ref{ass:bootstrap_precision}, Theorem \ref{thm:bootstrap} holds when the CCW are estimated using pooled logistic or complementary log-log regression.
\end{cor}



\section{Numerical study} \label{sec:simulation}

\subsection{Data generation process}
We consider a longitudinal setting with visits at time $t \in \{0,\dots,12\}$, with time-unit being \emph{month} for the sake of the simulation setting. We consider a two dimensional baseline covariates \(M=(M^1,M^2)\) with $M_1$ being continuous and $M_2$ being binary, which can be understood as being included in $X_0$, to align with the notation of the previous setting. Likewise, we consider a two dimensional time-varying covariates \(X_t=(X_{t}^1,X_{t}^2)\) with $X^1_t$ continuous and $X^2_t$ binary. We consider a binary treatment \(A_t\in\{0,1\}\). 
We model that $Y_t$ depends on $\overline A_t$ only through the sum $|\overline A_t|$ which corresponds to the total number of months under treatment. We refer to Appendix~\ref{app:simulation} for a precise description of the data generating process and to
Figure~\ref{fig:dag_longitudinal_schematic} for a DAG of the related random variables. \medskip

\begin{figure}[ht]
\centering
\begin{tikzpicture}[
    >=Stealth,
    every node/.style={font=\scriptsize},
    box/.style={
      draw,
      rounded corners=2pt,
      minimum width=1.15cm,
      minimum height=0.58cm,
      align=center
    },
    base/.style={box, fill=gray!15},
    xone/.style={box, fill=blue!12},
    xtwo/.style={box, fill=green!12},
    trt/.style={box, fill=orange!18},
    evt/.style={box, fill=red!10},
    hist/.style={box, fill=purple!25},
    lab/.style={font=\scriptsize},
    note/.style={font=\scriptsize, align=left}
]

\node[base] (M1) at (0, 1.2) {$M_1$};
\node[base] (M2) at (0, 0.0) {$M_2$};

\node[lab] at (3.0, 2.45) {visit $t-1$};
\node[lab] at (6.3, 2.45) {visit $t$};
\node[lab] at (9.6, 2.45) {visit $t+1$};

\draw[gray!70, ->] (1.5,2.1) -- (10.8,2.1);
\node[gray!70, anchor=west] at (10.9,2.1) {time};

\node[xone] (X1km1) at (3.0, 1.2) {$X_{1,t-1}$};
\node[xtwo] (X2km1) at (3.0, 0.0) {$X_{2,t-1}$};
\node[trt]  (Akm1)  at (3.0,-1.2) {$A_{t-1}$};
\node[hist] (Hkm1)  at (3.0,-2.2) {$|A_{t-1}|$};
\node[evt]  (Ykm1)  at (3.0,-3.2) {$Y_{t-1}$};

\node[xone] (X1k) at (6.3, 1.2) {$X_{1,t}$};
\node[xtwo] (X2k) at (6.3, 0.0) {$X_{2,t}$};
\node[trt]  (Ak)  at (6.3,-1.2) {$A_t$};
\node[hist] (Hk)  at (6.3,-2.2) {$|A_t|$};
\node[evt]  (Yk)  at (6.3,-3.2) {$Y_t$};

\node[xone] (X1kp1) at (9.6, 1.2) {$X_{1,t+1}$};
\node[xtwo] (X2kp1) at (9.6, 0.0) {$X_{2,t+1}$};
\node[trt]  (Akp1)  at (9.6,-1.2) {$A_{t+1}$};
\node[hist] (Hkp1)  at (9.6,-2.2) {$|A_{t+1}|$};
\node[evt]  (Ykp1)  at (9.6,-3.2) {$Y_{t+1}$};

\node at (1.7,-1.2) {$\cdots$};
\node at (1.7, 1.2) {$\cdots$};
\node at (1.7, 0.0) {$\cdots$};
\node at (1.7,-2.2) {$\cdots$};
\node at (1.7,-3.2) {$\cdots$};


\draw[->] (M1) to[bend left=10] (X1km1);
\draw[->] (M1) to[bend left=15] (Akm1);
\draw[->] (M1) to[bend left=22] (Ykm1);

\draw[->] (M2) to[bend right=5] (X2km1);
\draw[->] (M2) to[bend right=10] (Akm1);
\draw[->] (M2) to[bend right=16] (Ykm1);

\draw[->, gray!60] (M1) to[bend left=18] (Ak);
\draw[->, gray!60] (M2) to[bend right=12] (Ak);
\draw[->, gray!50] (M1) to[bend left=26] (Yk);
\draw[->, gray!50] (M2) to[bend right=20] (Yk);

\draw[->, gray!45] (M1) to[bend left=28] (Akp1);
\draw[->, gray!45] (M2) to[bend right=24] (Akp1);
\draw[->, gray!35] (M1) to[bend left=33] (Ykp1);
\draw[->, gray!35] (M2) to[bend right=28] (Ykp1);

\draw[->] (X1km1) to[bend left=48] (Akm1);
\draw[->] (X2km1) to[bend right=48] (Akm1);
\draw[->] (Akm1)  -- (Hkm1);
\draw[->] (X1km1) to[bend left=48] (Ykm1);
\draw[->] (X2km1) to[bend right=48] (Ykm1);
\draw[->] (Akm1)  to[bend right=48] (Ykm1);
\draw[->] (Hkm1)  -- (Ykm1);

\draw[->] (X1k) to[bend left=48] (Ak);
\draw[->] (X2k) to[bend right=48] (Ak);
\draw[->] (Ak)  -- (Hk);
\draw[->] (X1k) to[bend left=48] (Yk);
\draw[->] (X2k) to[bend right=48] (Yk);
\draw[->] (Ak)  to[bend right=48] (Yk);
\draw[->] (Hk)  -- (Yk);

\draw[->] (X1kp1) to[bend left=48] (Akp1);
\draw[->] (X2kp1) to[bend right=48] (Akp1);
\draw[->] (Akp1)  -- (Hkp1);
\draw[->] (X1kp1) to[bend left=48] (Ykp1);
\draw[->] (X2kp1) to[bend right=48] (Ykp1);
\draw[->] (Akp1)  to[bend right=48] (Ykp1);
\draw[->] (Hkp1)  -- (Ykp1);

\draw[->] (X1km1) -- (X1k);
\draw[->] (X2km1) -- (X2k);
\draw[->] (Akm1) -- (X1k);
\draw[->] (Akm1) -- (X2k);
\draw[->] (Akm1) -- (Ak);
\draw[->] (Hkm1) -- (Hk);

\draw[->] (X1k) -- (X1kp1);
\draw[->] (X2k) -- (X2kp1);
\draw[->] (Ak) -- (X1kp1);
\draw[->] (Ak) -- (X2kp1);
\draw[->] (Ak) -- (Akp1);
\draw[->] (Hk) -- (Hkp1);

\draw[->] (Yk) -- (Ykp1);
\draw[->] (Ykm1) -- (Yk);

\draw[->, dashed] (Akm1) to[bend left=12] (X1kp1);
\draw[->, dashed] (Akm1) to[bend left=8]  (X2kp1);
\draw[->, dashed] (Akm1) to[bend right=14] (Akp1);

\node[note, anchor=west] at (11.2,1.2) {time-varying\\continuous covariate};
\node[note, anchor=west] at (11.2,0.0) {time-varying\\binary covariate};
\node[note, anchor=west] at (11.2,-1.2) {treatment};
\node[note, anchor=west] at (11.2,-2.2) {cumulative\\treatment history};
\node[note, anchor=west] at (11.2,-3.2) {event indicator\\on current interval};

\node[note, anchor=west] at (-0.4,-4.5) {
\begin{tabular}{ll}
\textcolor{black}{\tikz\draw[->] (0,0)--(0.35,0);} & solid arrows: contemporaneous and lag-1 effects \\
\textcolor{black}{\tikz\draw[->, dashed] (0,0)--(0.35,0);} & dashed arrows: lag-2 treatment effects 
\end{tabular}
};

\node[note, anchor=west] at (-0.4,-5.6) {%
\parbox{12.5cm}{\emph{Notes.} At visit \(t\), treatment \(A_t\) depends on baseline variables \(M\), current covariates \(X_{t}\), and the two most recent treatment values. The covariates \(X_{t}\) depend on their previous value and on the two most recent treatment values. The event indicator \(Y_k\) depends on \((M,X_{t},A_t)\) and on the cumulative treatment \(| A_t|\). 
}
};

\end{tikzpicture}
\caption{Schematic longitudinal causal graph for the observational data-generating mechanism. 
}
\label{fig:dag_longitudinal_schematic}
\end{figure}

We aim to compare the two treatments strategies \emph{treat the first two months} v.s. \emph{treat the first seven months.} We consider the four following pair of interventions:
\begin{enumerate}[noitemsep,topsep=0pt]
    \item {\bf Strict adherence}: an individual follows the intervention if they are treated exactly the first two (resp. seven) months. 
    \item {\bf Grace period (1 month)}:  an individual follows the intervention if they are treated the first two (resp. seven) months, plus or minus one months.
    \item {\bf Grace period (2 month)}: an individual follows the intervention if they are treated the first two (resp. seven) months, plus or minus two months (and at least one month).
    \item {\bf Discontinuation}: an individual follows the intervention if they are treated for exactly two (resp. seven) months in total, with one month of discontinuation allowed.
\end{enumerate}
We refer to Figure~\ref{fig:admissible-trajectories} for a visual representation of these different interventions.\medskip

\begin{figure}
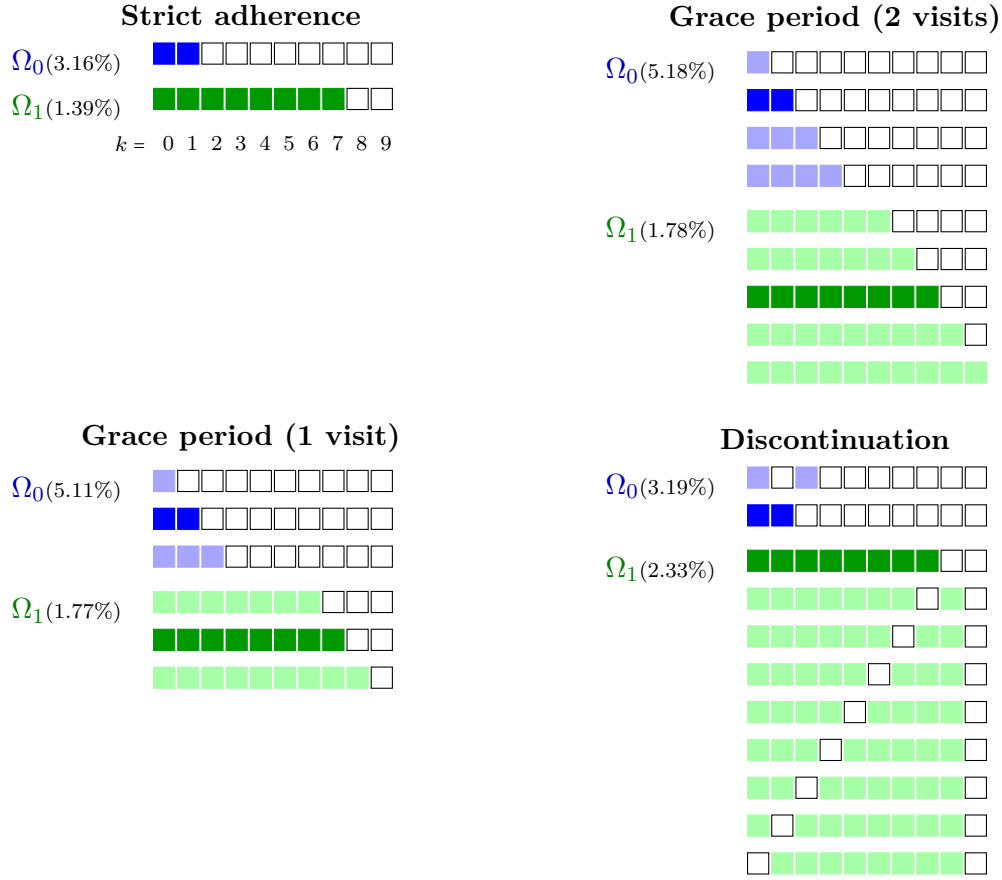

\centering

\begin{minipage}[t]{0.48\textwidth}
\vspace{0.3cm}
\centering

\textbf{Strict adherence}

\vspace{0.2cm}
\begin{tabular}{>{\centering\arraybackslash}m{1.45cm} m{4.2cm}}
\trajlineDOne{\DOnePct{3.16\%}}{1,1,0,0,0,0,0,0,0,0} 
\addlinespace[0.1cm]
\trajlineDTwo{\DTwoPct{1.39\%}}{1,1,1,1,1,1,1,1,0,0}
& \trajaxis \\[0.12cm]
\end{tabular}

\vspace{3.3cm}
\textbf{Grace period (1 visit)}

\vspace{0.2cm}
\begin{tabular}{>{\centering\arraybackslash}m{1.45cm} m{4.2cm}}
\trajlineDOneFaded{\DOnePct{5.11\%}}{1,0,0,0,0,0,0,0,0,0}
\trajlineDOne{}{1,1,0,0,0,0,0,0,0,0}
\trajlineDOneFaded{}{1,1,1,0,0,0,0,0,0,0}
\addlinespace[0.1cm]
\trajlineDTwoFaded{\DTwoPct{1.77\%}}{1,1,1,1,1,1,1,0,0,0}
\trajlineDTwo{}{1,1,1,1,1,1,1,1,0,0}
\trajlineDTwoFaded{}{1,1,1,1,1,1,1,1,1,0}
\end{tabular}
\end{minipage}
\hfill
\begin{minipage}[t]{0.48\textwidth}
\vspace{0.3cm}
\centering

\textbf{Grace period (2 visits)}

\vspace{0.2cm}
\begin{tabular}{>{\centering\arraybackslash}m{1.45cm} m{4.2cm}}
\trajlineDOneFaded{\DOnePct{5.18\%}}{1,0,0,0,0,0,0,0,0,0}
\trajlineDOne{}{1,1,0,0,0,0,0,0,0,0}
\trajlineDOneFaded{}{1,1,1,0,0,0,0,0,0,0}
\trajlineDOneFaded{}{1,1,1,1,0,0,0,0,0,0}
\addlinespace[0.1cm]
\trajlineDTwoFaded{\DTwoPct{1.78\%}}{1,1,1,1,1,1,0,0,0,0}
\trajlineDTwoFaded{}{1,1,1,1,1,1,1,0,0,0}
\trajlineDTwo{}{1,1,1,1,1,1,1,1,0,0}
\trajlineDTwoFaded{}{1,1,1,1,1,1,1,1,1,0}
\trajlineDTwoFaded{}{1,1,1,1,1,1,1,1,1,1}
\end{tabular}

\vspace{0.35cm}
\textbf{Discontinuation}

\vspace{0.2cm}
\begin{tabular}{>{\centering\arraybackslash}m{1.45cm} m{4.2cm}}
\trajlineDOneFaded{\DOnePct{3.19\%}}{1,0,1,0,0,0,0,0,0,0}
\trajlineDOne{}{1,1,0,0,0,0,0,0,0,0}
\addlinespace[0.1cm]
\trajlineDTwo{\DTwoPct{2.33\%}}{1,1,1,1,1,1,1,1,0,0}
\trajlineDTwoFaded{}{1,1,1,1,1,1,1,0,1,0}
\trajlineDTwoFaded{}{1,1,1,1,1,1,0,1,1,0}
\trajlineDTwoFaded{}{1,1,1,1,1,0,1,1,1,0}
\trajlineDTwoFaded{}{1,1,1,1,0,1,1,1,1,0}
\trajlineDTwoFaded{}{1,1,1,0,1,1,1,1,1,0}
\trajlineDTwoFaded{}{1,1,0,1,1,1,1,1,1,0}
\trajlineDTwoFaded{}{1,0,1,1,1,1,1,1,1,0}
\trajlineDTwoFaded{}{0,1,1,1,1,1,1,1,1,0}
\end{tabular}
\end{minipage}
\caption{Admissible treatment trajectories under the four different adherence scenarii. Filled squares indicate treatment intake ($A_t=1$), and empty squares indicate that no treatment was taken ($A_t=0$). Trajectories in $\Omega_0$ are shown in blue and trajectories in $\Omega_1$ in green. Strict trajectories are displayed with full opacity, whereas other admissible trajectories are displayed in a lighter shade. Percentages indicate the proportion of individuals in the observed data whose treatment trajectory belongs to the corresponding admissible set.} 
\label{fig:admissible-trajectories}
\end{figure}

A dataset of \(n=15{,}000\) observations is generated from the data generating process. We apply, for each of the four pairs of interventions described above, the CCW methodology, and get an estimate $\hat \Delta$ of the RMST constrast $\Delta$ between the two arms at time $\tau=12$:

$$
\Delta := \rho^{q^{\rm nat}_{\Omega_1}}(\tau=12) -   \rho^{q^{\rm nat}_{\Omega_0}}(\tau=12) 
$$

The process is then replicated over 500 Monte-Carlo (the number of replications is chosen based on computational resources time). Results are displayed in Figure~\ref{fig:monte_carlo}. The solid blue line in Figure~\ref{fig:monte_carlo} indicated the ground truth treatment effect under strict adherence. 
By contrast, the estimates obtained under the three other adherence settings are not centered around the value. This is expected, as these estimators no longer target the same causal estimand but the effect of a different stochastic intervention, as proven in Theorem~\ref{thm:consistency}. As one can note, in this setting, adding a grace period of one month does not change much the effect of the considered intervention. \medskip

For each admissible strategy, we also computed the associated ground truth (also by Monte Carlo computation) under the corresponding stochastic intervention. As shown in Figure~\ref{fig:monte_carlo} (Left), the estimates are centered around their own strategy-specific truth, confirming that the discrepancy with the strict adherence target does not reflect estimation failure, but rather a change in the estimand induced by relaxing the treatment adherence. At last, for a one Monte Carlo replicate and for each admissible strategy, 500 bootstrap samples were generated. The resulting estimates are displayed in Figure~\ref{fig:monte_carlo} (Right). As one can see, the bootstrap boxplot results reproduces precisely the Monte-Carlo boxplot results on the left panel, nicely illustrating the theoretical results presented in Section~\ref{sec:bootstrap}. This is further emphasized in the next subsection.

\begin{figure}[ht]
    \centering
    \includegraphics[width=0.49\linewidth]{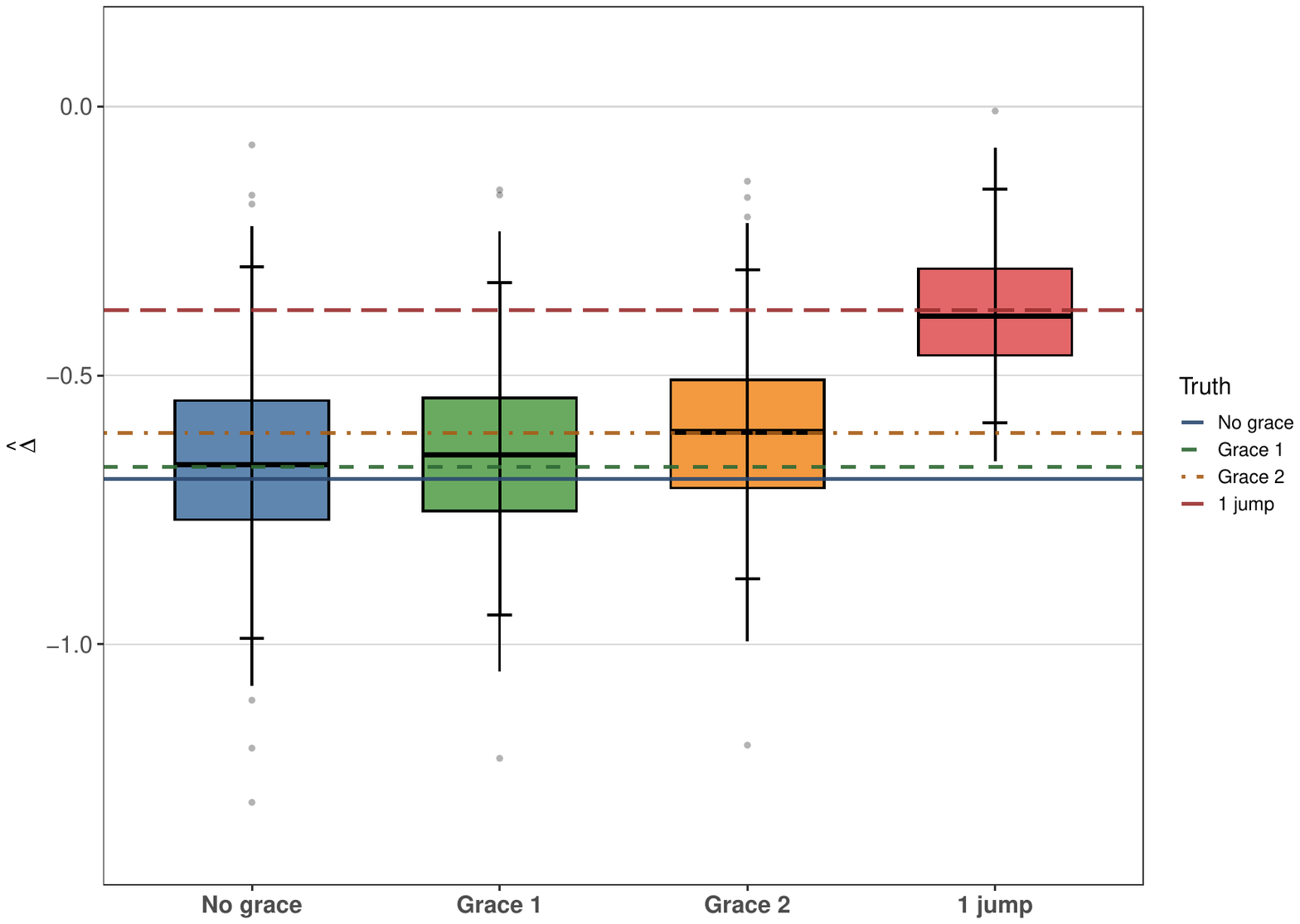}
    \includegraphics[width=0.49\linewidth]{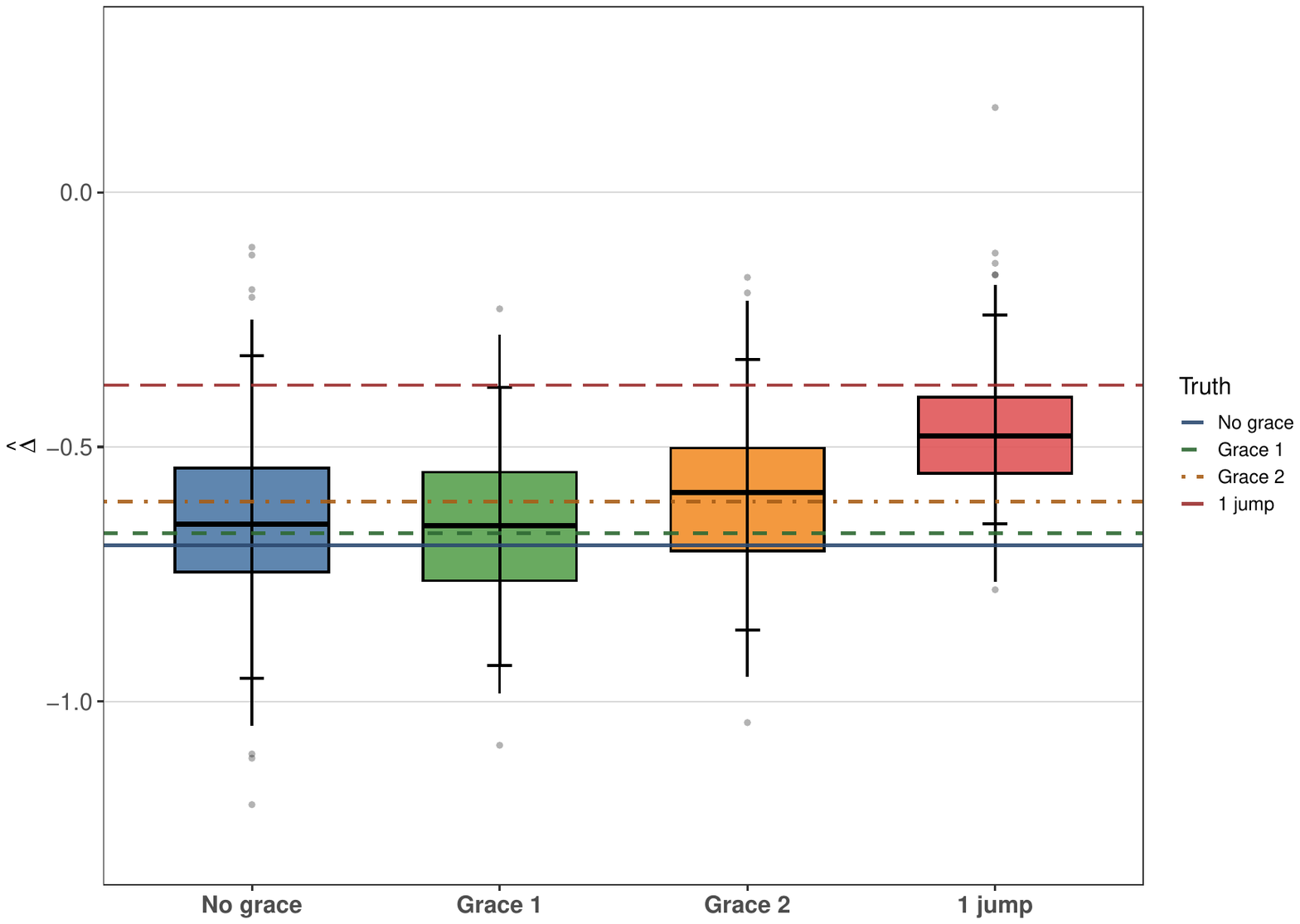}
    \caption{(Left) Boxplots of Monte-Carlo distributions of the CCW estimators obtained from 500 datasets under the different admissible treatment strategies. The horizontal reference lines represent the corresponding strategy-specific truths. (Right) Boxplots of bootstrapped CCW estimators obtained from a single dataset under the four adherence scenarios. For each scenario, 500 bootstrap samples were generated.}
    \label{fig:monte_carlo}
\end{figure}


\subsection{Bootstrap-based approximation of the sampling distribution}

We now turn to the evaluation of bootstrap-based methods for each admissible intervention. To this end, we generated 100 Monte Carlo replicates of the previous dataset and, for each replicate, performed 500 bootstrap resamples. For each Monte Carlo replicate and each admissible strategy, we constructed a 95\% percentile bootstrap confidence interval from the empirical 2.5\% and 97.5\% quantiles of the 500 bootstrap estimates, and recorded whether this interval contained the corresponding ground truth. Figure~\ref{fig:coverage_bootstrap_admissible} shows that the bootstrap confidence intervals achieve an empirical coverage close to the nominal 95\% level, with observed coverages around 97\% for grace periods of one and two time units, and around 90\% when one delayed treatment jump is allowed.

\begin{figure}[ht]
    \centering

    \subfloat[Grace period = 1]{
        \includegraphics[width=0.48\linewidth]{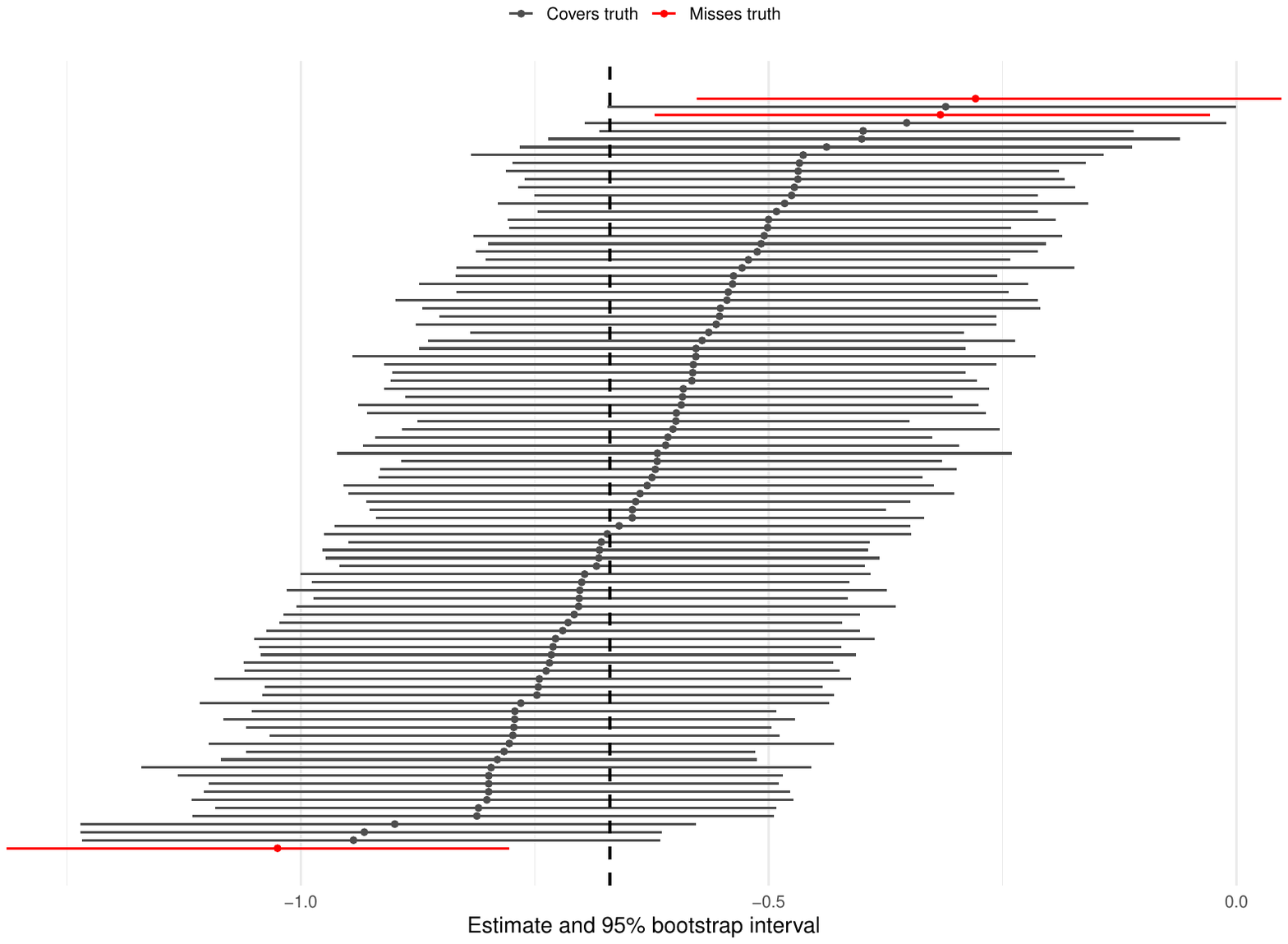}
    }
    \hfill
    \subfloat[Grace period = 2]{
        \includegraphics[width=0.48\linewidth]{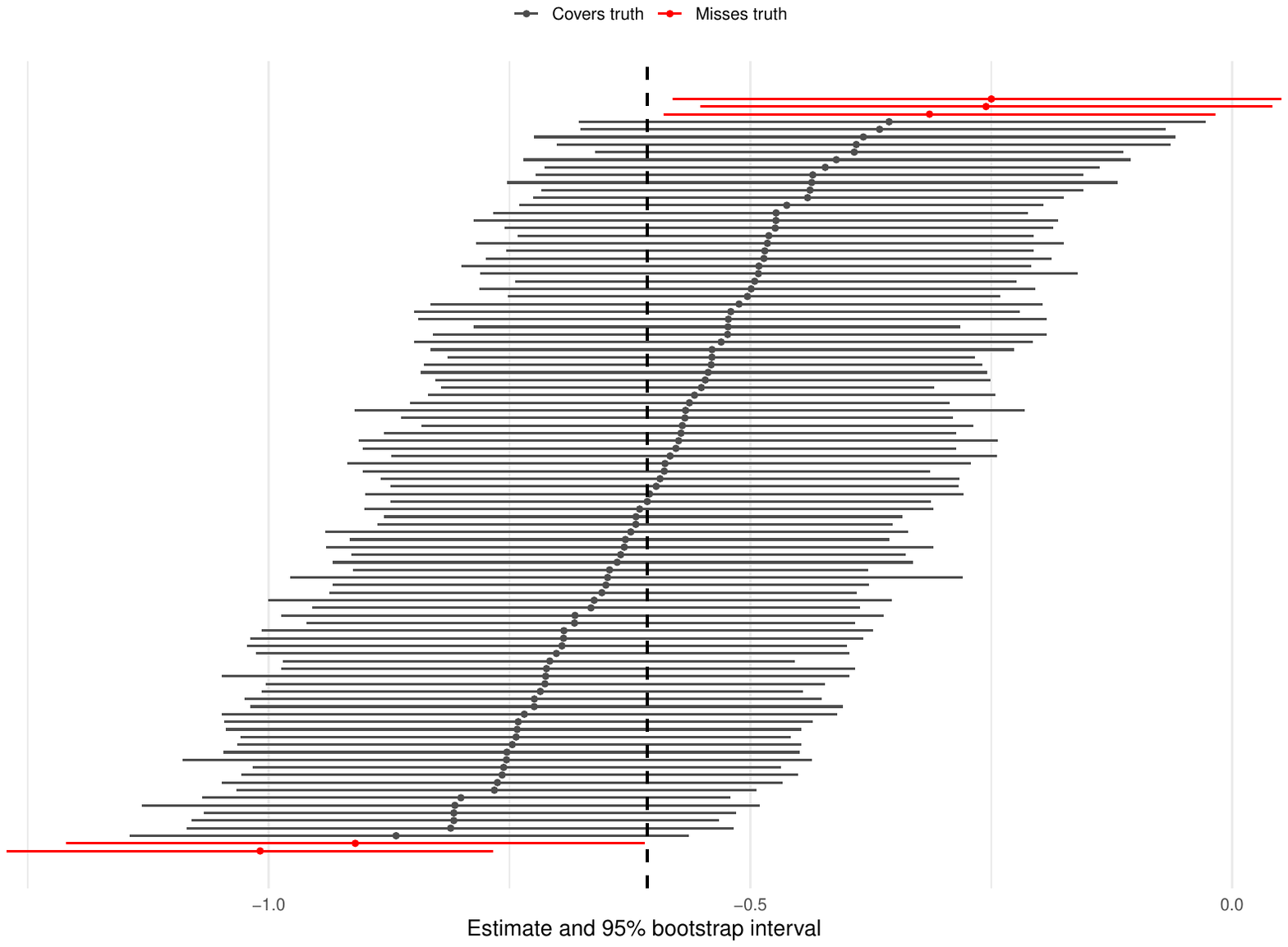}
    }
    \hfill
    \subfloat[Delayed treatment jumps allowed]{
        \includegraphics[width=0.48\linewidth]{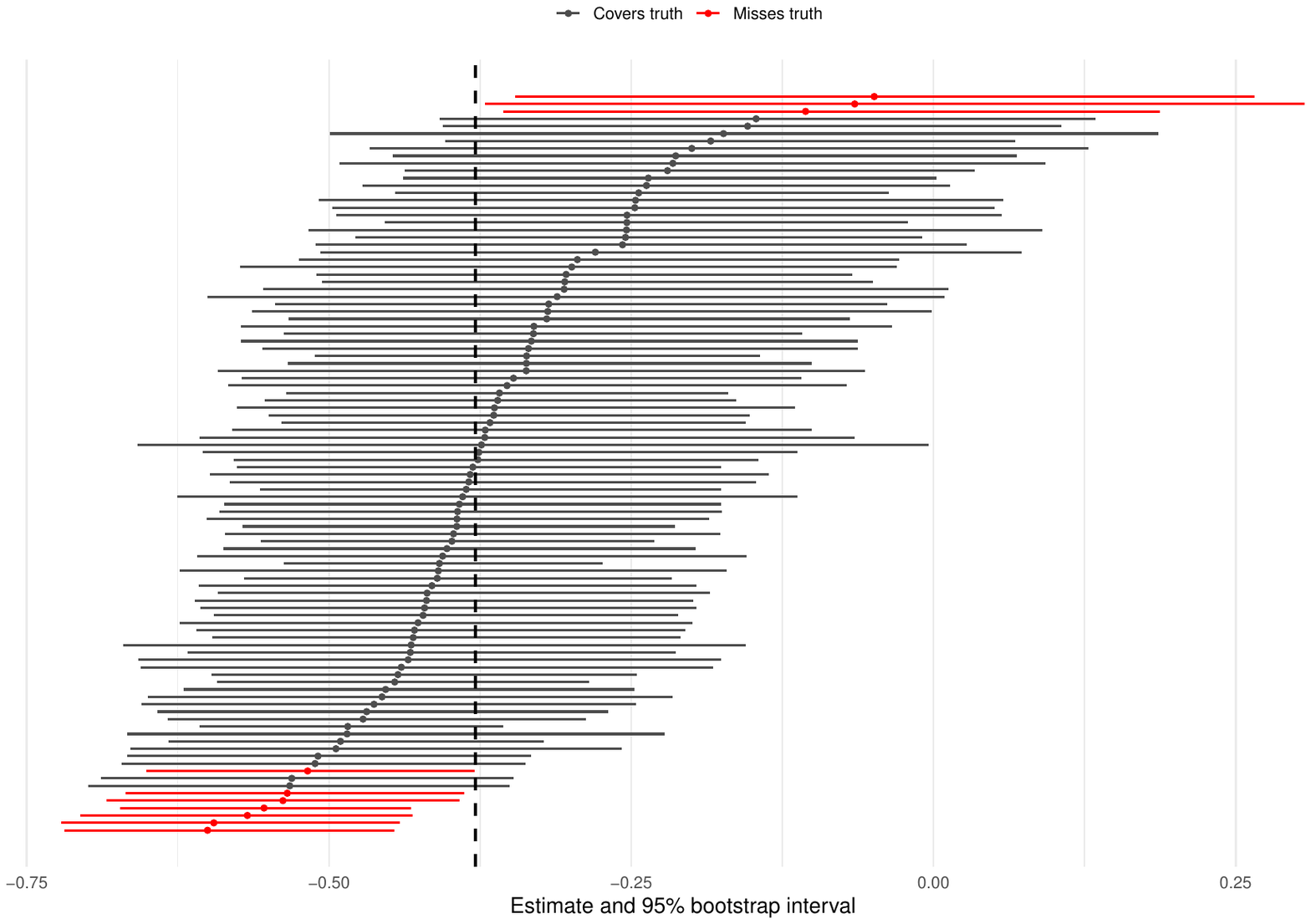}
    }

    \caption{
    Empirical coverage of the 95\% percentile bootstrap confidence intervals across 100 Monte Carlo replicates, based on 500 bootstrap resamples within each replicate, under increasingly flexible admissible treatment strategies.
    }

    \label{fig:coverage_bootstrap_admissible}
\end{figure}



We also compare the empirical distribution of the estimator across Monte Carlo replicates with a re-centered bootstrap distribution. The Monte Carlo distribution was obtained from 100 independent datasets replicates. For each Monte Carlo replicate, we obtain an estimate $\hat{\Delta}$ of the true parameter $\Delta_0$ and for each of these estimates, we also obtain a bootstrap estimates $\hat\Delta_{\rm b}^*$ of $\Delta_0$ which distribution is asymptotically centered at $\hat \Delta$. We recenter the distribution of the difference $\hat\Delta_{\rm b}^* - \hat\Delta$ at the true value of the parameter $\Delta_0$ that we estimate by averaging the $\hat \Delta$ over all Monte-Carlo simulations. The resulting centering writes
\[
\hat\Delta_{{\rm b},\mathrm{cent}}^*
=
\hat\Delta_{\rm b}^* - \hat\Delta + \bar{\hat\Delta}_{\mathrm{MC}},
\]
where \(\bar{\hat\Delta}_{\mathrm{MC}}\) is the mean of the Monte Carlo estimates across all replicates.  Figure~\ref{fig:MC_bootstrap} displays the Monte Carlo distribution together with this centered bootstrap distribution for each admissible intervention. As shown in Figure~\ref{fig:MC_bootstrap}, the centered bootstrap density closely tracks the Monte Carlo density, suggesting that the bootstrap provides a good approximation of the sampling distribution of the estimator in this setting. 


 

\begin{figure}[ht!]
    \centering

    \subfloat[Grace period = 1]{
        \includegraphics[width=0.48\linewidth]{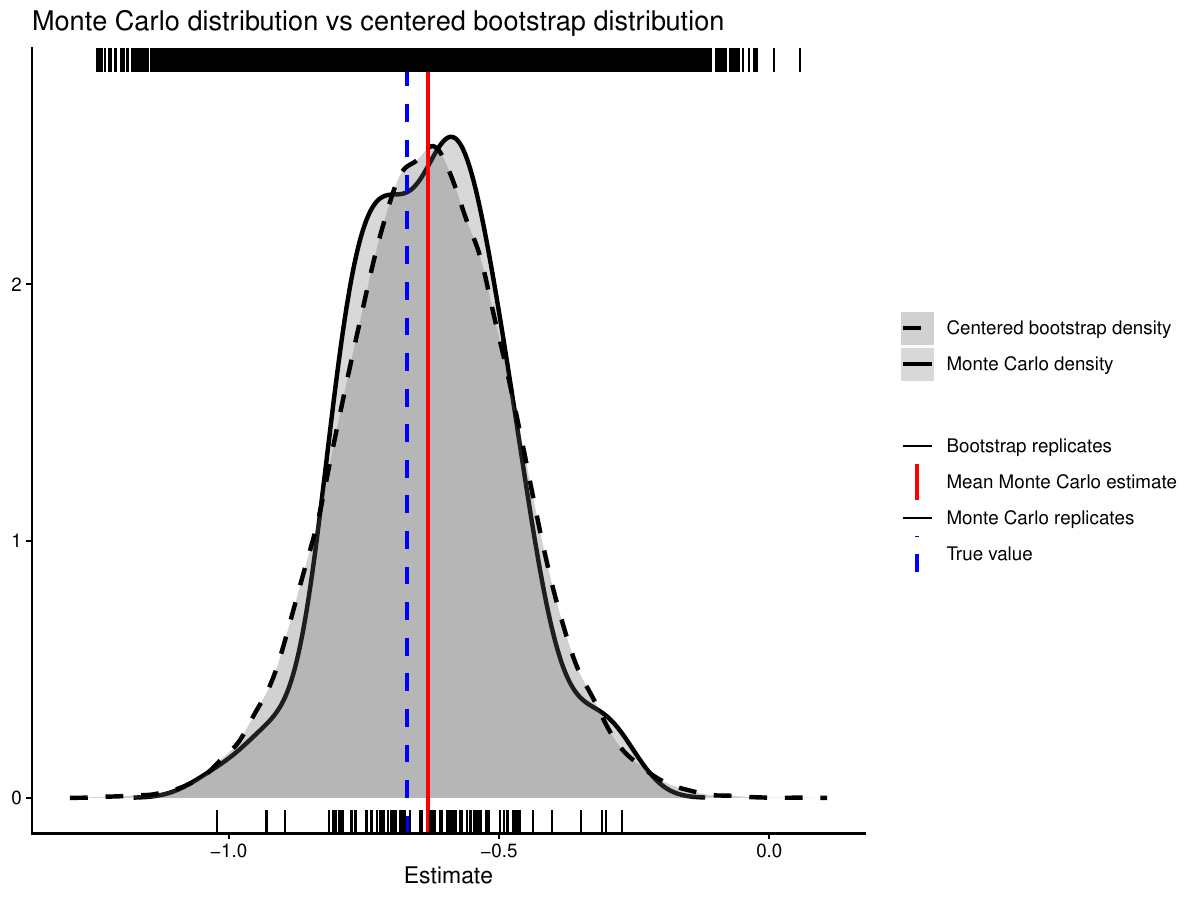}
    }
    \hfill
    \subfloat[Grace period = 2]{
        \includegraphics[width=0.48\linewidth]{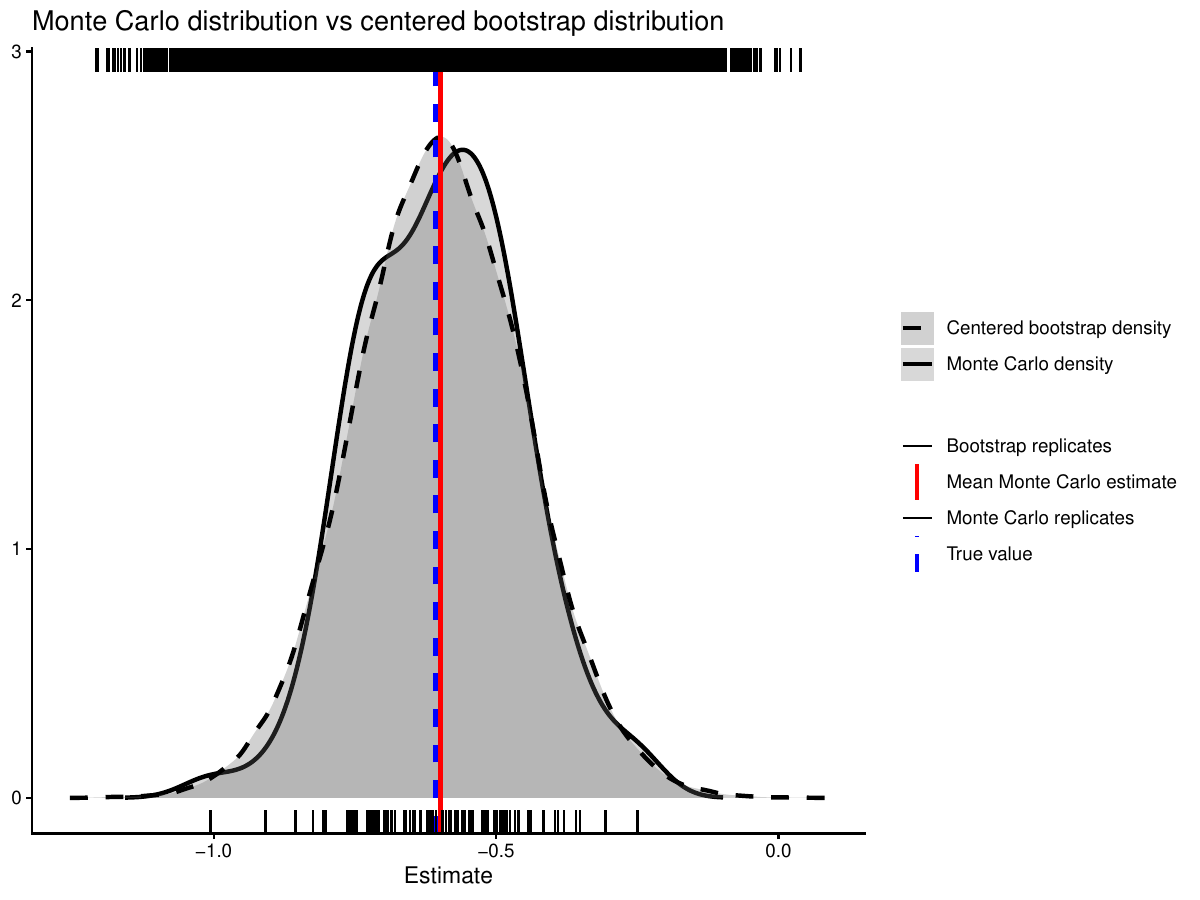}
    }
    \hfill
    \subfloat[Delayed treatment jumps allowed]{
        \includegraphics[width=0.48\linewidth]{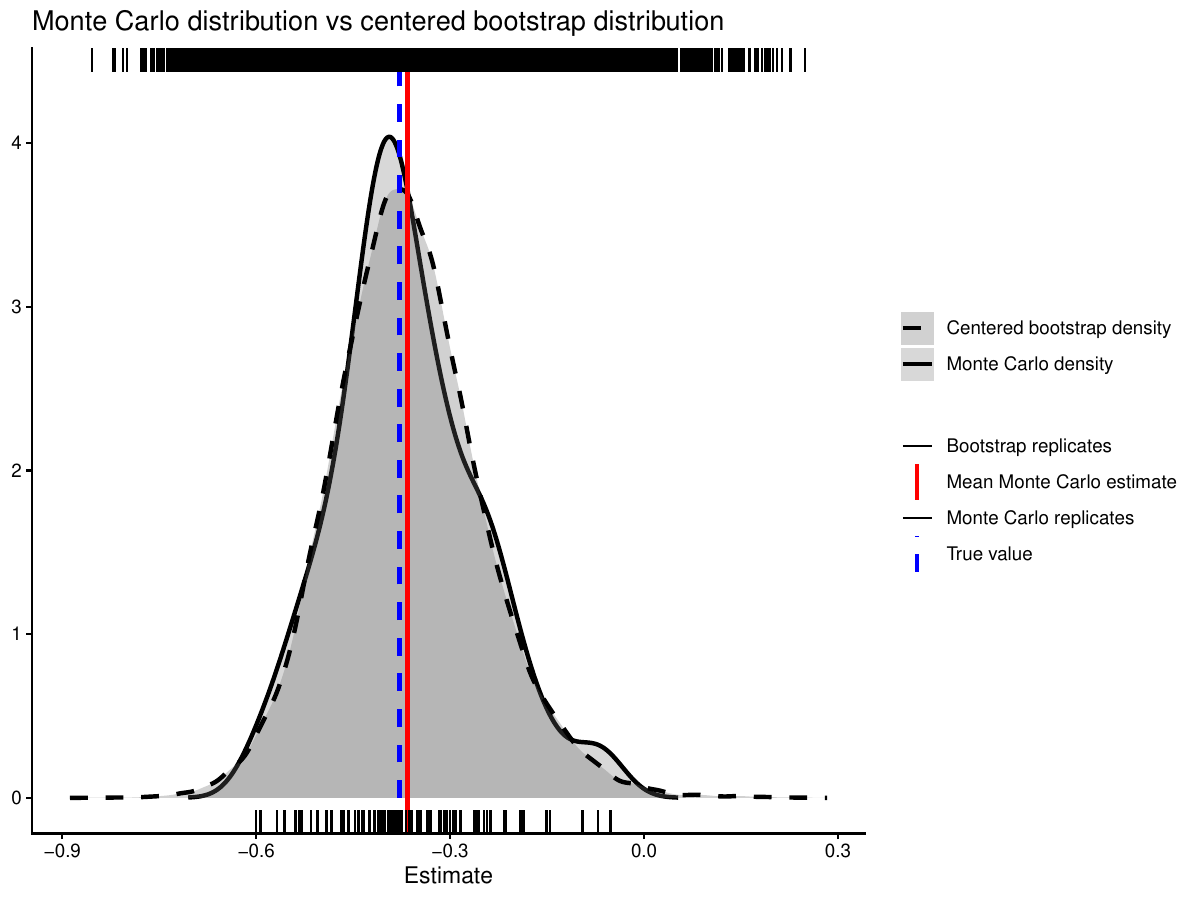}
    }

    \caption{
Distribution of the bootstrap variance estimates across 100 Monte Carlo replicates, based on 500 bootstrap resamples within each replicate, under increasingly flexible admissible treatment strategies.
    }

    \label{fig:MC_bootstrap}
\end{figure}

\section{Discussion}

Several directions for future research naturally emerge. A first avenue is the development of sensitivity analyses that quantify how estimation depends on the choice of admissible trajectory sets. For instance, when grace periods are introduced to address practical positivity constraints, their length is often chosen heuristically. The framework developed here suggests a systematic approach: define multiple admissible sets corresponding to different grace periods on the same dataset, each inducing a distinct estimand via its associated treatment mechanism. A key open question is then how to compare the proximity of the  resulting estimates to a scientifically meaningful target parameter. \medskip

A second direction is to extend this framework to more complex settings, such as dynamic treatment regimes with time-dependent confounding, where notions like grace periods become ill-defined. This arises, for example, in transplantation contexts involving waiting lists. In such settings, target trial emulation must address immortal time bias while accounting for the stochastic availability of treatment. One possible approach is to consider a sequence of nested target trials, as described by \cite{garcia2017value}, initiated at times when treatment becomes available. Within this framework, artificial censoring and weighting would still be required. However, cloning could be replaced by a form of comparability construction (e.g., matching) between treated individuals and those remaining at risk, based on time-varying confounders governing treatment availability. The nested sequence of trials would require a repeated construction of comparison sets which would induce a dependence structure distinct from that arising under cloning, as baseline populations would be redefined in a manner driven by the evolving confounding process. \medskip

The clone-censor-weight approach is widely used for causal analysis in target trial emulation with longitudinal time-to-event data, yet its statistical interpretation remains unclear. This work clarifies the estimand targeted by CCW by showing how it depends on the treatment mechanism and adherence definition. We provide theoretical guarantees for CCW estimators, including consistency and bootstrap inference under standard assumptions, and highlight how design choices such as the definition of admissible trajectories directly shape the causal estimand.

\section{Disclosure}

CV was a Sanofi employee at the time the work was done and may hold shares and/or stock options in the company. This work was
supported by Sanofi and INRIA (Institut National de
Recherche en Informatique et en Automatique) through the CIFRE program.

\section*{Acknowledgement} We are grateful to Raphaël Porchet for his time and remarks, which helped improve the manuscript.


\nocite{*}
\bibliographystyle{apalike}
\bibliography{ref}

\appendix

\section{Appendix to Section~\ref{sec:setting}} \label{sec:examples}

In this section, we present a list of examples formalized with the definitions of Section~\ref{subsec:notations}.

%
    




\begin{ex}[Example~\ref{ex:surgery}-\ref{ex:surgery-cont}, continued]\textbf{Surgery with grace period.} 

\noindent We describe the control intervention of \cite{maringe2020reflection}. The control population is represented by patient that never get the surgery, so that
\begin{align*}
	\Omega_{0,t} := 
		\{\overline 0_t\}  
\end{align*}
The associated natural kernel is degenerated with a Dirac distribution that puts all the mass at 0 for all times $t > 0$ 
$$q^{\rm nat}_{\Omega_0, t}(a  \mid H_{t}) = \ind\{a=0\}.$$
\end{ex}

\begin{ex}\textbf{Early intubation strategy.}

\noindent \cite{wanis2023emulating} emulate a target trial based on ICU database of patient with respiratory distress. The target trial compares early intubation ($\leq 8$ hours after admission) vs late intubation ($\geq 48$ hours after admission). For the “early intubation” strategy, admissible trajectories are given by 
\begin{align*}
	\Omega_{1,t} := 
	\begin{cases}
		\{0,1\}^t & \text{if } t \leq 7, \\
		\{\overline{a_t}~|~ |\overline{a_8}|\geq 1\},   & \text{if } t \geq 8. 
	\end{cases}
\end{align*}

In this setting, the natural kernel forces intubation if it has not occurred at the end of the 8 hours following entry in ICU. As there is no specification on the duration of the intubation in the trial, the kernel allows for either continuing or stopping the intubation once it is started. 
\begin{align*}
    q_{\Omega_1,t}^{\rm nat}(a \mid H_t)  = 
	\begin{cases}
		 \bbP(A_t=a \mid H_{t}, Y_{t}=C_{t}=0)  &\text{if } t \leq 7 \text{ or } |\overline{A}_{t-1}| \geq 1, \\[6pt]
      \ind\{a=1\}  &\text{if } t = 8 \text{ and } |\overline{A_7}| = 0.
	\end{cases}
\end{align*} 


Regarding the comparative  strategy “delayed intubation”, adherence requires that patients not be intubated within the first 48 hours, after which the physician may decide whether to proceed with intubation. Consequently, this strategy includes both patients intubated after 48 hours and those never intubated. The set of admissible trajectories is therefore given by    
\begin{align*}
	\Omega_{0,t} := \begin{cases}
 \{\overline 0_t\} \quad &\text{if} \quad t \leq 47
 \\
 \{\overline{a_t}~|~ |\overline{a_{47}}| = 0\} \quad &\text{if} \quad t > 47.
 \end{cases}
\end{align*}
The corresponding natural intervention is given by 
\begin{align*}
    q_{\Omega_0,t}^{\rm nat}(a \mid H_t)  = 
	\begin{cases}
      \ind\{a=0\}  &\text{if } t \leq  47, \\
      \bbP(A_t=a \mid H_{t}, Y_t = C_t = 0)  &\text{if } t \geq 48.
	\end{cases}
\end{align*} 
\end{ex}

\begin{ex}\textbf{Dynamic treatment strategy with grace period.} \label{ex:cd4}
    
    \noindent \cite{cain2010start} extend the work of \cite{hernan2006}, on when to start combined antiretroviral therapy as a treatment for HIV, considering cell count being below a threshold $x$ taking values from 200 and 500 in increment of 10. 
    Let $X_t \in [0,\infty)$ denote the CD4 cell count at time $t$. Treatment strategies correspond to initiating treatment within $m$ months after $X_t$ falls below a threshold $x$. 
    Thus, there are as many treatment strategies as there are thresholds $x$. Let $U_x := \inf \{ t : X_t \leq x \}$ denote the first time the CD4 cell count falls below $x$. 
    The set of  admissible trajectories is given by 
$$
\Omega_t := \{
		\overline{a_t}~|~ a_{U_x + m + s} = 1 \: \forall s \text{ s.t. } U_x + m + s \le t \},
$$
and the natural intervention by
\begin{align*}
	q_{\Omega,t}^{\rm nat}(a \mid H_t)  = 
	\begin{cases}
		\bbP(A_t=a\mid  H_t, Y_t=C_t=0) \quad &\text{ if } t < U_x + m, \\
		  \ind\{a=1\}   &\text{ if } t \geq U_x + m.   
	\end{cases}
\end{align*}
Similarly to Example~\ref{ex:surgery}, these treatment strategies involve a grace period, here defined as a time window of $m$ months. In contrast to previous examples, individuals may remain compatible with multiple treatment strategies throughout the entire follow-up. Indeed, an individual whose CD4 count decreases from 273 to 198 and who initiates treatment within $m$ months of falling to 198 will be compatible with all strategies corresponding to thresholds $x \in \{200, 210, \ldots, 270\}$. Such an individual remains cloned across eight strategies over the entire observation period.

\end{ex}

\begin{ex}\textbf{(Dis)continuation strategies.}

\noindent \cite{amiot2024stop} evaluate the optimal duration of immune checkpoint inhibitor (ICI) administration through four emulated trials, each comparing the causal effect of discontinuing versus continuing ICI at a specific time point $t^* \in \{6, 12, 18, 24\}$ (time index varies on months scale). The analysis is on patients who are still receiving treatment and exhibit disease control at time $t^*$. The treatment strategies can be formulated as follows: at the time of disease control $t^*$, decide whether to discontinue or continue treatment. A grace period of $m$ months is introduced to emulate a trial in which participants would be randomized either to discontinue treatment within three months or to continue treatment for at least an additional three months. The set of admissible trajectories is
\[
\Omega_t := \left\{
\overline{a}_t \; \middle| \; 
\overline{a}_{t^*} = \overline 1_{t^*}  \text{ and } \exists s \in \{t^*+1, \dots, t^* + m \} \text{ s.t } \ a_{s } = a_{s+ 1} = a_{s + 2} = 0
\right\}. 
\]
The corresponding natural intervention enforces treatment continuation up to time $t^*$, after which the strategy allows a grace period of $m$ months to initiate a 3-month \textit{discontinuation}. 
Letting 
$$
U := \inf\{s \geq t^* \mid A_{s}=A_{s+1}=A_{s+2}=0\},
$$
the natural intervention reads
\begin{align*}
	q_{\Omega,t}^{\rm nat}(a \mid H_t)  =  
	\begin{cases}
    \ind\{a=1\}   &\text{ if } t \leq t*, \\
	\bbP(A_t=a\mid H_t, Y_t=C_t=0)	&\text{ if } t \leq t^* + m - 3,\\
    \bbP(A_t=a\mid H_t, Y_t=C_t=0)	&\text{ if } t \geq t^* + m - 2 \text{ and } U \leq t^*+m-3,
    \\
    \ind\{a=0\}   &\text{ otherwise}. 
	\end{cases}
\end{align*}


\end{ex}

\section{Appendix to Section~\ref{sec:consistency}} \label{app:consistency}

\begin{proof}[Proof of Theorem~\ref{thm:consistency}]
If $\bbP(Y_{s-1}(q)=0)=0$, then $N_t^q$ is zero almost-surely and $\hat S^q_{\rm CCW}(t) = 0$ by convention, so that there is nothing to prove. Else, we write
$$
S^{q}(t)
= \prod_{s \leq t} \bbP(Y_s(q)=0 \mid Y_{s-1}(q)=0)
= \prod_{s \leq t}
\left(
1 - \frac{\bbP(Y_s(q)=1, Y_{s-1}(q)=0)}{\bbP(Y_{s-1}(q)=0)}
\right).
$$
By the law of large number, it is thus enough to show that
$$   
\bbE[W^q_{t-1} R_t]
= \bbP(Y_{t-1}(q)=0)
\quad \text{and} \quad 
\bbE[W^q_{t-1} R_t E_t]
= \bbP(Y_t(q)=1, Y_{t-1}(q)=0).
$$
Applying Lem~\ref{lem:identification} for $\psi(Y_t) = 1$ and $\psi(Y_t) = \ind\{Y_t=1\}$ ends the proof.
\end{proof}
\begin{lem}\label{lem:identification}
For any bounded measurable function $\psi$, it holds 
\[
\bbE[W_{t-1}^q R_t \psi(Y_t)] = \bbE[\ind\{Y_{t-1}(q)=0\} \psi(Y_t(q))].
\]
\end{lem}

\begin{proof} We define
    \begin{align*}
u_s(h_{s}) &:= \bbE\left[ \ind\{Y_{t-1}(q)=0\} \psi(Y_t(q)) \mid H_s(q) = h_s\right], \quad \text{and} \\
v_s(h_s) &:=  \bbE\left[\prod_{u=s}^{t-1} \frac{q_u(A_u\mid H_u)}{p_u^A(A_u\mid H_u)} \frac{\ind\{C_{u+1}=0\}}{p_{u+1}^C(H_u,A_u)} \ind\{Y_{t-1}=0\} \psi(Y_t) \mid H_s = h_s, C_s=0\right].
\end{align*}
We show the following by reverse induction on $s \leq t$:
\begin{equation}
    u_s(h_s) = v_s(h_s) \quad \text{for $h_s$ in  the support of $H_s(q)$.}
    \label{eq:rec}
\end{equation}
Fot $s=t$, we have, since $y_s = \ind\{x_s = \aleph\}$, 
$$
v_t(h_t)= \bbE\left[\ind\{Y_{t-1}=0\} \psi(Y_t) \mid H_t = h_t, C_t=0\right] = \ind\{y_{t-1}=0\} \psi(y_t). 
$$
Likewise,
$$
u_t(h_{t}) = \bbE\left[ \ind\{Y_{t-1}(q)=0\} \psi(Y_t(q)) \mid H_t(q) = h_t\right] = \ind\{y_{t-1}=0\} \psi(y_t).
$$
Now assume that Eq~\eqref{eq:rec} holds for some $s+1 \leq t$, and we show that it holds for $s$. Take $h_s$ in the support of $H_s(q)$. If $x_s = \aleph$, then both $v_s(h_s)$ and $u_s(h_s)$ are $0$ and there is nothing to prove. Otherwise, the event $H_s=h_s$ contains $Y_s=0$ and we can freely condition with the latter. 
It holds, by the induction hypothesis:
\begin{align*}
v_s(h_s) &=  \bbE\left[\frac{q_s(A_s\mid h_s)}{p_s^A(A_s\mid h_s)} \frac{\ind\{C_{s+1}=0\}}{p_{s+1}^C(h_s,A_s)} u_{s+1}(H_{s+1}) ~\middle|~ H_s = h_s, Y_s=C_s=0\right] \\
&= 
\sum_{a} q_s(a\mid h_s) \bbE\left[\frac{\ind\{C_{s+1}=0\}}{p_{s+1}^C(h_s,a)} u_{s+1}(H_{s+1}) ~\middle|~ A_s = a, H_s = h_s, Y_s=C_s=0\right],
\end{align*}
where we used that $p_s^A(a\mid h_s) = \bbP(A_s=a \mid H_s= h_s, Y_s=C_s=0)$. By Assumption~\ref{ass:censoring}, $C_{s+1}$ conditional to $A_s=a,H_s=h_s,Y_s=C_s=0$ is independent of $X_{s+1}$ hence
\begin{align*}
    \bbE&\left[\frac{\ind\{C_{s+1}=0\}}{p_{s+1}^C(h_s,a)} u_{s+1}(H_{s+1}) ~\middle|~ A_s = a, H_s = h_s, Y_s=C_s=0\right] \\
    &\quad\quad= \bbE\left[u_{s+1}(H_{s+1}) ~\middle|~ A_s = a, H_s = h_s, Y_s=C_s=0\right].
\end{align*}
By Assumptions~\ref{ass:consistency} and \ref{ass:ignorability}, for any $a$ such that $q_s(a \mid h_s)>0$, it holds
$$
\bbE\left[u_{s+1}(H_{s+1}) ~\middle|~ A_s = a, H_s = h_s, Y_s=C_s=0\right] = \bbE\left[u_{s+1}(H_{s+1}(q)) ~\middle|~ A_s(q) = a, H_s(q) = h_s, Y_s(q)=C_s(q)=0\right].
$$
Again, one can remove $Y_s(q)=0$ because of the conditioning $H_s(q)=h_s$, and we then can remove $C_s(q)=0$ by  Assumption~\ref{ass:censoring}. In the end, we find
\begin{align*}
v_s(h_s) = \sum_a q_s(a\mid h_s) \bbE\left[u_{s+1}(H_{s+1}(q)) ~\middle|~ A_s(q) = a, H_s(q) = h_s\right] = u_s(h_s),
\end{align*}
which ends the induction. Now, since $X_0(q)$ and $X_0$ have the same law, we find that
\begin{align*}
\bbE[\ind\{Y_{t-1}(q)=0\}\psi(Y_t(q))] &= \bbE[u_0(X_0(q))] \\
&= \bbE[v_0(X_0)] \\
&= \bbE\left[\prod_{u=0}^{t-1} \frac{q_u(A_u\mid H_u)}{p_u^A(A_u\mid H_u)} \frac{\ind\{C_{u+1}=0\}}{p_{u+1}^C(H_u,A_u)} \ind\{Y_{t-1}=0\} \psi(Y_t)\right] \\
&= \bbE[W_{t-1}^q R_t \psi(Y_t)],
\end{align*}
which ends the proof.

\end{proof}

\section{Appendix to Section~\ref{sec:bootstrap}} \label{app:bootstrap}

In this appendix, we provide precise technical conditions ensuring the validity of the bootstrap of a $Z$-estimator approach to CCW. We do not give the most general conditions but try to strike a balance between generality and readability. 

We first observe that if the hazards and the CCW weights admits estimators that are asymptotically normal with weakly consistent bootstrap versions, then the corresponding results for the difference in survival probabilities at a given time and for the RMST follow from the $\Delta$-method. 
Second, we note that the hazards satisfy estimating equations and that the weights are functions of certain conditional probability distributions. 
Therefore, if these conditional probability distributions are estimated using a Z-estimator, the full set of parameters can be estimated using a stacked Z-estimator. Then classical results on Z-estimators implie the weak bootstrap consistency of the hazard and conditional probability distributions estimators. the corresponding properties for the contrasts is, as explained above, obtained via the $\Delta$-method.

\subsection{Contrasts and hazards}
We first show that the hazards and the weights can be estimated through a Z-estimation approach under suitable assumptions.

We fix two set $\Omega_0$ and $\Omega_1$ of admissible trajectories and consider the natural interventions $q^{{\rm nat}}_{\Omega_0}$ and $q^{{\rm nat}}_{\Omega_1}$ respectively associated with $\Omega_0$ and $\Omega_1$. For legibility, we simply write with a superscript $j$ the quantities associated with $q^{{\rm nat}}_{\Omega_j}$. We set 
 \[
 h^j_s=\bbP(Y^j_s=1 \mid Y_{s-1}^j=0)= \frac{\bbP(Y_s^j=1, Y_{s-1}^j=0)}{\bbP(Y_{s-1}^j=0)}
 \]
the discrete hazard under regime  $q^{j}$, $j=0,1$. We set  
\[
h^j=(h^j_1,\dots,h^j_\tau)^\transpose\in\mathbb R^{\tau}.
\]
Note that this is equivalent as seeing the hazard as a function on the finite time axis $1,\ldots,\tau$ defined by
\[
h^j(t)=\sum_{s=1}^\tau h^j_s \ind\{t=s\}.
\]
The probability of survival at time $t$ under regime $j$ is
\[
S^{j}(t)=\prod_{s=1}^{t}(1-h^j_s),
\]
and the difference of risk of survival at time $t$ is
\[
\Delta(t)=S^{1}(t)-S^{0}(t).
\]
We define the restricted mean survival time (RMST) up to time $t$ under regime $j$ by
\[
\mu^{j}(t)=\sum_{s=1}^{t} S^{j}(s),
\qquad j\in\{0,1\},
\]
The difference of RMST  up to time $t$ is given by
\[
\Delta_{\mathrm{RMST}}(t)=\mu^{1}(t)-\mu^{0}(t).
\]

\subsection{Estimating equation}\label{subsec:estimating_equation}

It follows from Lemma \ref{lem:identification} applied once with $\psi(Y_s) = 1$ and $\psi(Y_s)=\ind\{Y_s=1\}$ that 
\[
h^j_s= \frac{\bbE[W^j_{s-1} R_s E_s]}{\bbE[W^j_{s-1} R_s]}.
\]
which yields
\begin{equation*}
\bbE[W^j_{s-1} R_s (E_s-h_s^j)]=0. 
\end{equation*}

Hence to estimate $h^j$, it is sufficient to be able to estimate the $W_s^j$ for all $1 \leq s \leq \tau$. In view of Corollary \ref{cor:natural_weights}, this in turn  reduces to estimate the quantities $p^{\Omega_j}_s(X_s, H_{s-1})$ and $p^{C}_{s+1}(H_{s})$. 

We assume that these conditional probabilities belong to finite-dimensional parametric families
\[
p_s^{\Omega_j}(H_{s})=p_s^{\Omega_j}(H_{s};\eta^{\Omega_j}),
\]
and
\[
p_{s+1}^{C}(A_s,H_s)=p_{s+1}^{C}(A_s, H_s;\eta^{C}),
\]
where
\[
\eta^{\Omega_j}\in \Theta^{\Omega_j}\subset \R^{d_{\Omega_j}},
\qquad
\eta^{C}\in \Theta^{C}\subset \R^{d_{C}}.
\]
Define the stacked nuisance parameter
\[
\eta=\Big(\eta^{\Omega_0},\eta^{\Omega_1},\eta^{C}\Big) \in \Theta^{\Omega_0}  \times \Theta^{\Omega_1} \times  \Theta^{C}
\]
which collects all parameters indexing the artificial and natural censoring models.

We assume that the true nuisance parameter $\eta_0$ is characterized as the unique solution of the stacked estimating equation
\[
\mathbb E\big[\psi^{\rm nuis}(O,\eta_0)\big] = 0,
\]
where
\[
\psi^{\rm nuis}(O;\eta)
=
\Big(
\psi^{\Omega_0}(O;\eta^{\Omega_0}),\psi^{\Omega_1}(O;\eta^{\Omega_1}),
\psi^{C}(O;\eta^{C})
\Big).
\]
We define the full parameter vector
\[
\theta=(\eta,h^0,h^1)\in\Theta\subset \R^d \; \textnormal{with} \; d=d_{\Omega_0}+d_{\Omega_1}+d_{C}+2\tau,
\]
where 
\[
h^j=(h^j_1,\dots,h^j_{\tau})^\transpose\in \R^\tau,\qquad j=0,1.
\]
are the hazards associated to the set of admissible trajectories $\Omega^j$ and $\eta$ is the stacked nuisance parameter.
Let $\psi(O;\theta)\in\mathbb R^d$ denote the stacked estimating function,
\begin{equation}\label{eq:critical_curve}
\psi(O;\theta)=
\begin{pmatrix}
\psi^{\rm nuis}(O;\eta)\\
\psi^0_{1}(O;\eta,h^0_1)\\
\vdots\\
\psi^0_{\tau}(O;\eta,h^0_{\tau})\\
\psi^1_{1}(O;\eta,h^1_{1})\\
\vdots\\
\psi^1_{\tau}(O;\eta,h^1_{\tau})
\end{pmatrix},
\end{equation}
where
\[
\psi_{j,s}(O;\eta,h^j_s)
=
W^j_{s-1}(\eta)\,R_{s}^j\,(E_{s}-h^j_s),
\qquad j=0,1,\;s=1,\dots,\tau.
\]

It follows that the true parameter $\theta_0$ is characterized as the unique solution of the stacked estimating equation
\[
\Psi_\bbP(\theta)=\bbE(\psi(O;\theta))=0
\]
Let $\bbP_n$ denote the empirical measure, $\bbE_{\bbP_n} f = \frac1n\sum_{i=1}^n f(O_i)$. The empirical counterpart of $\Psi_\bbP$ is
\[
\Psi_n(\theta):=\bbE_{\bbP_n} \psi(O;\theta).
\]
Let $\hat\theta_n$ be an estimator obtained by approximately solving the stacked estimating equation
\[
\Psi_n(\theta)=0.
\]
The required approximation rate is stated in Assumption \ref{ass:root}. Writing
\[
\hat\theta_n=(\hat\eta_n,\hat h_n^0,\hat h_n^1),
\]
the estimator $\hat\eta_n$ is the Z-estimator of the nuisance parameter, and
\[
\hat h_n^j=(\hat h_{n,1}^j,\dots,\hat h_{n,\tau}^j)^\top,\qquad j=0,1,
\]
is the corresponding estimator of the hazard vector. 

We now introduce 
\begin{equation*}
g_{\rm RD}(t,\theta)=\prod_{s=1}^{t}(1-h^1_s)-\prod_{k=1}^{t}(1-h^0_s)\quad \textnormal{and} \quad g_{\mathrm{RMST}}(t,\theta)
=
\sum_{k=0}^{t}\prod_{s=1}^{k}(1-h^1_s)
-
\sum_{k=0}^{t}\prod_{s=1}^{k}(1-h^0_s),
\end{equation*}

We obtain respectively estimator of  $\Delta(t)$ and $\Delta_{\mathrm{RMST}}(t)$ by setting

\begin{equation*}
\hat{\Delta}(t)=g_{\mathrm{RD}}(t,\hat{\theta})\quad \textnormal{and} \quad \hat{\Delta}_{\mathrm{RMST}}(t)=g_{\mathrm{RMST}}(t,\hat{\theta}).
\end{equation*}

\subsection{Assumption} \label{app:bootstrap_assumption}

For legibility, we shorten $\psi(O;\theta)$ in $\psi_\theta$.

\paragraph{Assumptions for consistency and asymptotic normality.}
Let $\psi_\theta(O)$ denote the stacked estimating function defined by equation \eqref{eq:critical_curve}, and recall the following notations
\[
\Psi_\bbP(\theta)=\bbE_\bbP(\psi_\theta),
\qquad
\Psi_n(\theta)=\bbE_{\mathbb P_n}(\psi_\theta).
\]

\begin{enumerate}[label=(Z\arabic*)]
\item \label{ass:Zidentifiability} \textbf{Identification.}
There exists $\theta_0\in \mathring\Theta$ such that $\Psi_\bbP(\theta_0)=0$, and for any sequence
$\theta_n\in\Theta$,
\[
\Psi_\bbP(\theta_n)\to 0
\quad\Longrightarrow\quad
\theta_n\to\theta_0.
\]

\item\label{ass:GC} \textbf{Global Glivenko-Cantelli property.}
The class
$
\{\psi_\theta:\theta\in\Theta\}
$
is $\bbP$-Glivenko-Cantelli:
\[
\sup_{\theta\in\Theta}\|\Psi_n(\theta)-\Psi_\bbP(\theta)\|\to 0 \quad \text{almost-surely}.
\]


\item\label{ass:root} \textbf{Approximate empirical root.}
The estimator $\hat\theta_n$ satisfies
\[
\Psi_n(\hat\theta_n)=o_{\bbP}(n^{-1/2}).
\]

\item \label{ass:derivatice} \textbf{Local smoothness and square-integrable derivative envelope.}
There exists a neighborhood $N(\theta_0)\subset\Theta$ such that, for $\bbP$-almost every $o \in \mathcal{O}$,
the map $\theta\mapsto\psi_\theta(o)$ is continuously differentiable on $N(\theta_0)$,
and there exists a measurable function $F(o)$ with
\[
\sup_{\theta\in N(\theta_0)}\|\dot\psi_\theta(o)\|\le F(o)
\quad \bbP-\text{a.s.},
\qquad
\bbE_\bbP F^2<\infty.
\]

\item \label{ass:square_in} \textbf{Square integrability at the truth.}
\[
\bbE_\bbP\|\psi_{\theta_0}\|^2<\infty.
\]

\item \label{ass:non_sing} \textbf{Differentiability and nonsingularity of the population Jacobian.}
The map $\Psi_\bbP$ is differentiable at $\theta_0$, with derivative matrix
\[
A=\dot\Psi_\bbP(\theta_0)=\bbE_\bbP \left( \dot\psi_{\theta_0} \right),
\]
and $A$ is nonsingular.
\end{enumerate}

\begin{rem}
    In practical setting, the weights are often estimated via pooled logistic regression. the assumption $(Z1)$ is a classical identifiability assumption (see \cite{vanderVaart1998}). In the case of pooled logistic regression, it can be obtained by assuming the population log-likelihood to be concave. In this setting, the assumption $(Z2)$ can be obtained by assuming that the parameter space is compact and the covariates are bounded. The aim of $(Z1)$ is to ensure the identifiability of the parameter and $(Z2)$ allows to prove the consistency of the estimator and its bootstraped version. If these results can be proved directly these assumptions can be bypassed. 
\end{rem}

\subsection{Consequence}

\begin{lem}\label{lem:Donsker}
Suppose there exists a neighborhood $N(\theta_0)\subset\Theta$ such that, for $\bbP$-almost every $o \in \mathcal{O}$,
the map
\[
\theta\mapsto \psi_\theta(o)
\]
is continuously differentiable on $N(\theta_0)$, that $\bbE_{\bbP} \|\psi_{\theta_0}\|^2 < \infty$, and that there exists a measurable function
$F$ satisfying
\[
\sup_{\theta\in N(\theta_0)}\|\dot\psi_\theta(o)\|
\le F(o)
\qquad \bbP-\text{a.s.},
\qquad
\bbE F^2<\infty.
\]
Then there exists $\eta>0$ such that:

\begin{enumerate}
\item[(i)] the class
\[
\mathcal F_\eta
=
\{\psi_\theta:\|\theta-\theta_0\|\le \eta\}
\]
is $\bbP$-Donsker;

\item[(ii)] 
\[
\bbE\|\psi_\theta-\psi_{\theta_0}\|^2\to 0
\qquad\text{as }\theta\to\theta_0.
\]
\end{enumerate}
\end{lem}

\begin{proof}
Since $N(\theta_0)$ is a neighborhood of $\theta_0$, it contains a ball $B(\theta_0,\eta)$ of radius $\eta>0$ centered in $\theta_0$

\noindent (i) Let $\theta_1,\theta_2$ in $B(\theta_0,\eta)$. Since
$\theta\mapsto\psi_\theta(o)$ is continuously differentiable on $B(\theta_0,\eta)$ for $\bbP$-almost
every $o \in \mathcal{O}$, the mean value inequality yields
\[
\|\psi_{\theta_1}(o)-\psi_{\theta_2}(o)\|
\le
\sup_{\theta\in B(\theta_0,\eta)}\|\dot\psi_\theta(o)\|\,\|\theta_1-\theta_2\|
\le
F(o)\,\|\theta_1-\theta_2\|
\]
almost surely. In particular, on $B(\theta_0,\eta)$, we have almost surely
\begin{equation}\label{ineq:loclip}
\|\psi_\theta(o)-\psi_{\theta_0}(o)\|
\le
F(o)\,\|\theta-\theta_0\|.
\end{equation}
Hence, the class
\[
\mathcal F_\eta
=
\{\psi_\theta:\|\theta-\theta_0\|\le \eta\}
\]
is a finite-dimensional parametric class that is locally Lipschitz in $\theta$ with
$L^2(\bbP)$ envelope $\eta F+ \|\psi_{\theta_0}\|$. Hence, by \cite[Example 19.7]{vanderVaart1998},
$\mathcal F_\eta$ is $\bbP$-Donsker.\\

\noindent (ii) Squaring Inequality \eqref{ineq:loclip} yields
\[
\bbE_\bbP\|\psi_\theta-\psi_{\theta_0}\|^2
\le
\|\theta-\theta_0\|^2\,\bbE_\bbP F^2.
\]
Since $\bbE_\bbP F^2<\infty$, it follows that
\[
\bbE_\bbP\|\psi_\theta-\psi_{\theta_0}\|^2\to 0
\qquad\text{as }\theta\to\theta_0.
\]
\end{proof}

\subsection{Bootstrap weak validity}

\subsubsection*{Bootstrap weak validity for the stacked Z-estimator}

We assume that we carry out the empirical bootstrap as follows

\begin{enumerate}[label=(B\arabic*)]
\item \label{ass:bootstrap_sampling}
Conditional on $\mathbb{O}_n$, the bootstrap sample $\mathbb{O}_n^\ast=(  O_1^*,\dots, O_n^*)$ is drawn by resampling baseline individuals with replacement.

\item \label{ass:bootstrap_procedure}
In each bootstrap sample, one recomputes the admissibility indicators, artificial censoring times, nuisance estimator, weighted hazards, and weighted Kaplan--Meier estimators by applying exactly the same stacked estimating-equation procedure to $\mathbb P_n^*$ where $\mathbb P_n^*$ is the bootstrap empirical distribution.

\item \label{ass:bootstrap_precision} Writing
\[
\Psi_n^*(\theta):=\bbE_{\mathbb P_n^*}\psi_\theta,
\]
the bootstrap estimator satisfies
\[
\Psi_n^*(\hat\theta_n^*)=o_{\mathbb P^*}(n^{-1/2}).
\]
\end{enumerate}

\begin{prp}[Bootstrap weak validity of the stacked estimator]\label{prop:Z_weakboot}
Under assumptions, \ref{ass:Zidentifiability}--\ref{ass:non_sing} and \ref{ass:bootstrap_sampling}--\ref{ass:bootstrap_precision}
\begin{enumerate}
\item 
\begin{equation}\label{eq:normstd}
\sqrt n(\hat\theta-\theta_0)\xrightarrow[n\to\infty]{\cL} \cN\left(0,\,A^{-1}\bbE_\bbP[\psi_{\theta_0}\psi_{\theta_0}^{\transpose}]A^{-\transpose}\right),
\end{equation}

\item the bootstrap estimator $\hat \theta^*_n$ is such that, for all bounded continuous function $\vp$,
\begin{equation}\label{eq:normboot}
\bbE[\vp(\sqrt n(\hat\theta^*-\hat\theta)) \mid \bbO_{n}] \xrightarrow[n\to\infty]{\bbP} \bbE[\vp(Z)] \quad \text{where} \quad Z \sim \cN\left(0,\,A^{-1} \bbE_\bbP[\psi_{\theta_0}\psi_{\theta_0}^{\transpose}]A^{-\transpose}\right),
\end{equation}
and where $\bbO_{n}=(O_{1},\dots,O_{n})$.
\end{enumerate}
\end{prp}

\begin{proof}
 Assumptions \ref{ass:Zidentifiability}, \ref{ass:GC}, \ref{ass:root}, \ref{ass:bootstrap_procedure}, \ref{ass:bootstrap_precision}, \ref{ass:square_in}, and \ref{ass:non_sing} imply conditions (A), (B), (D), (E) of \cite[Theorem 10.16]{Kosorok2008}. Condition (C) is verified thanks to assumption \ref{ass:derivatice} and \ref{ass:square_in} together with Lemma \ref{lem:Donsker}. Furthermore, even though $\Theta$ is not assumed open, the fact that $\theta_0$ is in the interior of $\Theta$ and the consistency allows to work on an open subset of $\Theta$. Therefore, Theorem 10.16 of Kosorok applies directly to the stacked Z-estimator and yields
\[
\sqrt n(\hat\theta-\theta_0)\xrightarrow[n\to\infty]{\cL}\cN \left(0,\,A^{-1}\bbE_\bbP[\psi_{\theta_0}\psi_{\theta_0}^{\transpose}]A^{-\transpose}\right).
\]
The same theorem also yields point $2.$
\end{proof}

\begin{cor}[Bootstrap weak validity for the fixed-time survival contrast]
\label{thm:appendix-main}
Under Assumptions \ref{ass:Zidentifiability}--\ref{ass:non_sing}, and \ref{ass:bootstrap_sampling}--\ref{ass:bootstrap_precision}, the following hold.

\begin{enumerate}
\item If
\[
\sigma_\Delta^2=\dot g_{\mathrm{RD}}(\theta_0)
A^{-1}\mathbb E_\bbP[\psi_{\theta_0}\psi_{\theta_0}^{\transpose}]A^{-\transpose}
\dot g_{\mathrm{RD}}(\theta_0)^{\transpose},
\]
then
\[
\sqrt n(\hat\Delta(t)-\Delta(t))\xrightarrow[n\to\infty]{\cL} \cN(0,\sigma_\Delta^2).
\]

\item The bootstrap estimator $\hat \Delta^*(t)$ is such that, for all bounded continuous function $\vp$,
\[
\bbE[\vp(\sqrt n(\hat\Delta^*(t)-\hat\Delta(t))) \mid \bbO_{n}] \xrightarrow[n\to\infty]{\bbP} \bbE[\vp(Z)] \quad \text{where} \quad Z \sim \cN(0,\sigma_{\Delta}^2),
\]
and where $\bbO_{n}=(O_{1},\dots,O_{n})$.
\end{enumerate}
\end{cor}

\begin{proof}
The map
\[
g_{\mathrm{RD}}(\theta)=\prod_{s=1}^{t}\{1-h_s^1\}-\prod_{s=1}^{t}\{1-h_s^0\}
\]
depends only on the hazard coordinates and is continuously differentiable on a neighborhood of $\theta_0$. 
%
Applying the ordinary delta method to formula \eqref{eq:normstd} give point $1.$
and applying the bootstrap delta method \cite[Theorem 23.9]{vanderVaart1998} to formula \eqref{eq:normboot} gives point $2.$ 
\end{proof}

One can similarly prove the following.

\begin{cor}
\label{thm:appendix-rmst}
Under assumptions \ref{ass:Zidentifiability}--\ref{ass:non_sing}, and \ref{ass:bootstrap_sampling}--\ref{ass:bootstrap_precision}, the following hold for
\[
\Delta_{\mathrm{RMST}}(t)=\mu^{1}(t)-\mu^{0}(t).
\]

\begin{enumerate}
\item  If
\[
\sigma_{\mathrm{RMST}}^2=
\dot g_{\mathrm{RMST}}(\theta_0)
A^{-1}\mathbb E_\bbP[\psi_{\theta_0}\psi_{\theta_0}^{\transpose}]A^{-\transpose}
\dot g_{\mathrm{RMST}}(\theta_0)^{\transpose},
\]
then
\[
\sqrt n\{\hat\Delta_{\mathrm{RMST}}(t)-\Delta_{\mathrm{RMST}}(t)\}
\xrightarrow[n\to\infty]{\cL} \cN(0,\sigma_{\mathrm{RMST}}^2),
\]

\item the bootstrap estimator $\hat \theta^*_n$ is such that, for all bounded continuous function $\vp$,
\[
\bbE[\vp(\sqrt n(\hat\Delta_{\mathrm{RMST}}^*(t)-\hat\Delta_{\mathrm{RMST}}(t))) \mid \bbO_{n}] \xrightarrow[n\to\infty]{\bbP} \bbE[\vp(Z)] \quad \text{where} \quad Z \sim \cN(0,\sigma_{\mathrm{RMST}}^2),
\]
and where $\bbO_{n}=(O_{1},\dots,O_{n})$.
\end{enumerate}
\end{cor}



\section{Pooled generalized linear models for the nuisance probabilities}
\label{app:pooled-glm}

The estimating-equation construction from Subsection \ref{subsec:estimating_equation} leaves the three nuisance-score
blocks \(\psi^{\Omega_0}\), \(\psi^{\Omega_1}\), and \(\psi^C\) unspecified.
We provide their expression when the conditional probabilities entering the
CCW weights are estimated by pooled Bernoulli generalized linear models.  This
includes the pooled logistic and pooled complementary log--log regressions
commonly used in CCW analyses.

\subsection{Responses, risk indicators, and regressors.}
We start by introducing auxiliary notations that provide homogeneous notations allowing to handle the three nuisance blocks simultaneously.

For the artificial censoring, we set, for \(j\in\{0,1\}\)
\begin{align*}
U^{\Omega_j}_{i,s}:=\Gamma_{\Omega_j,i,s}, \quad
V^{\Omega_j}_{i,s}:=R_{i,s}\ind\{Y_{i,s}=0\}\Gamma_{\Omega_j,i,s-1},
\quad
Z^{\Omega_j}_{i,s}
    &:=m_s^{\Omega_j}(H_{i,s})\in\R^{d_{\Omega_j}},
\end{align*}

The random variable \(V^{\Omega_j}_{i,s}\) selects individuals whose current treatment is
observed and who have remained compatible with \(\Omega_j\) through time
\(s-1\). The maps $m_s^{\Omega_j} \colon \mathcal{H}_s \to \R^{d_{\Omega_j}}$ are measurable feature functions.  At baseline we use the conventions
\(\Gamma_{\Omega_j,i,-1}=1\), \(Y_{i,0}=0\), and \(R_{i,0}=1\).
Thus
\begin{equation}
\label{eq:glm-artificial-probability}
p_s^{\Omega_j}(H_{i,s})
=\bbP\!\left(
    U^{\Omega_j}_{i,s}=1
    \mid H_{i,s},V^{\Omega_j}_{i,s}=1
\right).
\end{equation}

For the natural-censoring model, we set
\begin{align*}
U^C_{i,s}:=\ind\{C_{i,s+1}=0\}, \quad V^C_{i,s}:=R_{i,s}\ind\{Y_{i,s}=0\}, \quad Z^C_{i,s}:=m_s^C(A_{i,s},H_{i,s})\in\R^{d_C}.
\end{align*}
where the $m_s^{C} \colon \mathcal{A}_s \times \mathcal{H}_s \to \R^{d_C}$ are measurable feature functions. This yields
\begin{equation}
\label{eq:glm-natural-probability}
p^C_{s+1}(A_{i,s},H_{i,s})
=\bbP\!\left(U^C_{i,s}=1\mid A_{i,s},H_{i,s},V^C_{i,s}=1\right).
\end{equation}

\begin{rem}
If natural censoring is assumed independent and the corresponding factor in
the weights is known, the \(C\)-component below can simply be omitted.
\end{rem}

For \(s=0,\ldots,\tau-1\), we define the following sigma-algebras:
\begin{align*}
\mathcal F^A_{s}
&:=\sigma\!\left(
    \overline X_{s},
    \overline A_{s-1},
    (C_{r},Y_{r})_{r=1}^{s}
\right)
\label{eq:pre-treatment-filtration}
\end{align*}
with the convention
\[
\mathcal F^A_{0}=\sigma(X_{0}).
\]
We also define
\begin{align*}
\mathcal F^C_{s}
&:=
\mathcal F^A_{s}\vee\sigma(A_{s})                                      \notag\\
&=
\sigma\!\left(
    \overline X_{s},
    \overline A_{s},
    (C_{r},Y_{r})_{r=1}^{s}
\right).
\label{eq:pre-censoring-filtration}
\end{align*}
Thus,
\[
\mathcal F^A_{s} \subseteq \mathcal F^C_{s} \subseteq \mathcal F^A_{s+1}.
\]
For notational purpose, we set for \(j\in\{0,1\}\)
\[
\mathcal F^{\Omega_j}_{s}:=\mathcal F^A_{s}.
\]
To treat the three nuisance models simultaneously, write
\[
\mathcal K:=\{\Omega_0,\Omega_1,C\},
\]
and let \((U^k_{i,s},V^k_{i,s},Z^k_{i,s},\eta^k)\) denote the corresponding
quantities above for \(k\in\mathcal K\). 

We use the following restriction convention. Conditional on \(V^k_s=1\),
the probability space and all random variables are restricted to
\(\{V^k_s=1\}\), with law \(\bbP(\,\cdot\mid V^k_s=1)\); conditioning on
\(\mathcal F^k_s\) then means conditioning on its trace
\[
\mathcal F^k_s\big|_{\{V^k_s=1\}}
:=\{F\cap\{V^k_s=1\}:F\in\mathcal F^k_s\}.
\]
We retain the notation
\(\bbP(\,\cdot\mid\mathcal F^k_s,V^k_s=1)\), and use the analogous
convention for conditional expectations. If \(\bbP(V^k_s=1)=0\), the
corresponding score term is zero and is omitted. Equations
\eqref{eq:glm-artificial-probability} and
\eqref{eq:glm-natural-probability} can therefore be written
\begin{equation*}
\bbP \left(U_s^k=1 \mid\mathcal F^k_s,V^k_s=1\right).
\end{equation*}

\subsection{Pooled GLM specification}

We will now model the probability of remaining uncensored either artificially or naturally using generalized linear models (GLM). We know by assumption that the probability of being naturally uncensored at a time is strictly positive. Moreover, it is reasonable to assume that before the end of follow-up the probability of being naturally censored is strictly positive. It follows from Assumption \eqref{ass:positivity_treatment}  that the probability that the action trajectory remains admissible is bounded away from zero. On the other hand, it may be possible that under a given strategy at a given time step all the structurally possible actions are admissible. The probability for a given action to be admissible is then equal to one. This situation can be encountered in the presence of a grace period (see for instance Example  \ref{ex:cd4}). As logistic and complementary log-log regression are unable to model a probability equal to one, we introduce the following indicator to handle this apparent difficulty.

\begin{equation*}
D^{\Omega_j}_{i,s}:=\ind\!\left\{\mathcal S_s(H_{i,s})\not\subseteq\mathcal A_{\Omega_j,s}(H_{i,s})\right\}
\end{equation*}
where $\mathcal A_{\Omega_j,s}(h)$ is the set of admissible actions (see Definition \ref{def:qnat}) and  $\mathcal S_s(H_{i,s})=\{a \in \mathcal A_s \mid \bbP(A_s = a \mid H_{i,s},Y_{i,s}=C_{i,s}=0)>0 \}$ is the set of actions that are structurally possible after history $H_{i,s}$. We assume $\mathcal S_s(H_{i,s}$ to be known for every relevant history $(H_{i,s}$. Consequently, \(D^{\Omega_j}_{i,s}=0\) exactly when every structurally possible action is admissible, in which case \(U^{\Omega_j}_{i,s}=1\) with conditional probability one.
For notational convenience, we also set
\[
D^C_{i,s}:=1.
\]

For each \(k\in\mathcal K\), let \(G_k:(0,1)\to\R\) be a strictly
monotone bijective link whose inverse \(g_k=G_k^{-1}\) is twice
continuously differentiable.  

Pooled logistic regression uses
\[
G_k(u)=\log\!\left\{\frac{u}{1-u}\right\},
\qquad
g_k(z)=\frac{\exp(z)}{1+\exp(z)}.
\]
For pooled complementary log--log regression, the model is applied to the probability of remaining uncensored \(u\). In this setting,
\[
G_k(u)=\log\{-\log(u)\},
\qquad
g_k(z)=\exp\{-\exp(z)\}.
\]

The pooled GLM for component \(k\) is
\begin{align*}
G_k\!\left(\mu^k_{i,s}(\eta^k)\right)
    &=(Z^k_{i,s})^\transpose\eta^k,\\
\mu^k_{i,s}(\eta^k)
    &:=g_k\!\left((Z^k_{i,s})^\transpose\eta^k\right),\\
\widetilde\mu^k_{i,s}(\eta^k)
    &:=1-D^k_{i,s}+D^k_{i,s}\mu^k_{i,s}(\eta^k).
\label{eq:pooled-glm-model}
\end{align*}
Thus, \(\widetilde\mu^k_{i,s}=1\) on rows where the probability is known and equals the modeled probability \(\mu^k_{i,s}\) otherwise.
We assume that there exists \(\eta_0^k\) such that
\begin{equation}
\label{eq:pooled-glm-sequential-assumption}
\bbP\!\left(U^k_s=1\mid\mathcal F^k_s,V^k_s=1\right)
=\widetilde\mu^k_s(\eta_0^k) \quad \textnormal{for } k \in \mathcal{K}.
\end{equation}
Assuming that the nuisance component models are correctly specified at
\(\eta_0=(\eta_0^{\Omega_0},\eta_0^{\Omega_1},\eta_0^C)\),  it follows from equation \eqref{eq:pooled-glm-sequential-assumption} that,
\[
D^k_{s}\bbE\!\left[U^k_s-\mu^k_s(\eta_0^k)
\mid\mathcal F^k_s,V^k_s=1\right]=0.
\]

\subsection{Estimating equation}
The aim of this section is to set up the estimating equation for $\eta_0=(\eta_0^{\Omega_0},\eta_0^{\Omega_1},\eta_0^C)$. Under equation \eqref{eq:pooled-glm-sequential-assumption}, the contribution of individual \(i\) is
\begin{equation}
\label{eq:pooled-glm-loglikelihood}
\ell_i^k(\eta^k)
=\sum_s V^k_{i,s}D^k_{i,s}
\left[
U^k_{i,s}\log\mu^k_{i,s}(\eta^k)
+(1-U^k_{i,s})\log\{1-\mu^k_{i,s}(\eta^k)\}
\right].
\end{equation}
Define
\begin{equation*}
z^k_{i,s}(\eta^k):=(Z^k_{i,s})^\transpose\eta^k,
\qquad
b_k(z):=\frac{g_k'(z)}{g_k(z)\{1-g_k(z)\}}.
\end{equation*}
Differentiating \eqref{eq:pooled-glm-loglikelihood} gives the individual
nuisance estimating function:
\begin{equation}
\label{eq:pooled-glm-score}
\psi^k(O_i;\eta^k)
=\sum_s V^k_{i,s}D^k_{i,s} b_k\!\left(z^k_{i,s}(\eta^k)\right)
\{U^k_{i,s}-\mu^k_{i,s}(\eta^k)\}Z^k_{i,s}.
\end{equation}
Since \(V^k_s\) is binary, conditioning each summand on \(V^k_s=1\) and then
applying the tower property under the restricted law gives
\begin{align*}
\bbE\!\left[\psi^k(O;\eta_0^k)\right]
&=\sum_s\bbP(V^k_s=1)\,
\bbE\!\left[
D^k_{s} b_k\!\left(z^k_s(\eta_0^k)\right)
\{U^k_s-\mu^k_s(\eta_0^k)\}Z^k_s
\mid V^k_s=1
\right]\\
&=\sum_s\bbP(V^k_s=1)\,
\bbE\!\left[
D^k_{s} b_k\!\left(z^k_s(\eta_0^k)\right)
\bbE\!\left[
U^k_s-\mu^k_s(\eta_0^k)
\mid\mathcal F^k_s,V^k_s=1
\right]Z^k_s
\mid V^k_s=1
\right]\\
&=0
\end{align*}
for every \(k\in\mathcal K\).
Recall that the nuisance estimating function is the column vector
\begin{equation}
\label{eq:pooled-glm-stacked-score}
\psi^{\rm nuis}(O;\eta)
=
\begin{pmatrix}
\psi^{\Omega_0}(O;\eta^{\Omega_0})\\
\psi^{\Omega_1}(O;\eta^{\Omega_1})\\
\psi^C(O;\eta^C)
\end{pmatrix},
\qquad
\eta=(\eta^{\Omega_0},\eta^{\Omega_1},\eta^C).
\end{equation}
This completely specifies the estimating functions. Consequently, we assume that an estimator $\hat{\eta}_n$ of $\eta_0$ satisfies the following
\begin{equation}\label{eq:pooled-glm-approx-root}
\bbE_{\bbP_n} \left [\psi^{\rm nuis}(O;\hat{\eta}_n) \right ]
=o_{\bbP}(n^{-1/2}).
\end{equation}
\begin{rem}
  The preceding construction directly models the two probabilities
\(p_s^{\Omega_0}\) and \(p_s^{\Omega_1}\).  Alternatively, both may be
derived from one model for the observed treatment mechanism but this requires stronger assumptions.  
\end{rem}
\subsection{Sufficient assumptions}
We make the following assumptions which are standard in the context of generalized linear models and show that they imply assumptions (Z1), (Z2) and  (Z4)-(Z6). We further prove that the approximate nuisance-root condition
\eqref{eq:pooled-glm-approx-root}, together with exact solutions of the weighted hazard equations, implies (Z3).

\begin{enumerate}[label=(GLM\arabic*)]
\item \label{ass:glm1} There exists $\eta_0=(\eta_0^{\Omega_0},\eta_0^{\Omega_1},\eta_0^C)$ satisfying equation \eqref{eq:pooled-glm-sequential-assumption} and $\theta_0=(\eta_0,h^0_0,h^1_0)$ is an interior point of $\Theta$;
\item \label{ass:glm2} The parameter space $\Theta$ is compact; 
\item \label{ass:glm3} The functions $m_s^k$ for $s=0,\ldots,\tau-1$ and $k \in \mathcal{K}$ are bounded measurable functions;
\item \label{ass:glm4} For every \(k\in\mathcal K\), the inverse link
\(g_k:\R\to(0,1)\) is twice continuously differentiable and
\(\lvert g_k^\prime(t)\rvert>0\) for every \(t\in\R\);
\item \label{ass:glm5} The matrices $\bbE\left[\sum_sV_s^kD_s^kZ_s^k(Z_s^k)^\transpose
\right]$ are positive definite for $k \in \mathcal{K}$;
\item \label{ass:glm6} $\bbE[W^j_{s-1}R_s]>0
\qquad\text{for every }j\in\{0,1\},\ s=1,\ldots,\tau.$
\end{enumerate}
We first prove the smoothness of the population estimating map.
\begin{lem}\label{lem:continuity}
Under assumptions \ref{ass:glm2}--\ref{ass:glm4}, the map \(\Psi_\bbP\) is the restriction
to \(\Theta\) of a continuously differentiable map, and
\begin{equation}\label{eq:diff_under}
\dot\Psi_\bbP(\theta)
=\bbE\left[\dot\psi(O;\theta)\right].    
\end{equation}
\end{lem}
\begin{proof}
Let \(K\) be a compact neighbourhood of \(\Theta\) in the ambient parameter
space. By \ref{ass:glm3} and since the number of follow-up time is finite, the linear predictors $z^k_{i,s}(\eta^k)$
are uniformly bounded on \(K\). Thus, there exists \(M_K \in \R_+\) such that
\[
\sup_{\theta\in K}\max_{k,s}|z^k_{i,s}(\eta^k)|\le M_K,
\qquad I_K:=[-M_K,M_K].
\]
It follows from \ref{ass:glm4} that there exists $\varepsilon>0$ such that
\begin{equation}\label{eq:bound_proba}
\varepsilon\le\mu_s^k(\eta^k)\le 1-\varepsilon
\qquad(k\in\mathcal K, s, \theta \in K).
\end{equation}
Moreover, for \(k\in\mathcal K\) and \(u\in\{0,1\}\), the function
\[
q_{k,u}(t):=b_k(t)\{u-g_k(t)\}
\]
is continuously differentiable. Hence, $q_{k,u}$ and $q_{k,u}^\prime$ are bounded on \(I_K\). Since for every $k \in \mathcal{K}$
\[
\psi^k(O;\eta^k)
=\sum_sV_s^kD_s^kq_{k,U_s^k}\!\left((Z_s^k)^\transpose\eta^k\right)Z_s^k
\]
it follows that  the application $\eta^k \mapsto \psi^k(O;\eta^k)$ and its derivative are
continuous and thus bounded on \(K\).

Each weight \(W_{s-1}^j(\eta)\) is a finite product of indicators and
reciprocals of the full probabilities \(\widetilde\mu^k_s(\eta^k)\). These
probabilities are either one or, by \eqref{eq:bound_proba}, uniformly bounded
away from zero. Since they and their first
derivatives are continuous and uniformly bounded on \(K\), the product and
chain rules show that, there exists $M \in \R_+$ such that
\[
\sup_{\theta\in K}\max_{j,s}
\left\{
|W_{s-1}^j(\eta)|
+\left\|\frac{\partial W_{s-1}^j(\eta)}{\partial\eta}\right\|
\right\}
\le M
\]
for every realization. For the hazard blocks,
\begin{align*}
\frac{\partial\psi_{j,s}}{\partial\eta}
&=R_s(E_s-h_s^j)\frac{\partial W_{s-1}^j}{\partial\eta},
&
\frac{\partial\psi_{j,s}}{\partial h_r^\ell}
&=-\ind\{(\ell,r)=(j,s)\}W_{s-1}^jR_s.
\end{align*}

The functions \(R_s\), \(E_s\), and \(h_s^j\) are bounded on \(K\). It follows that \(\dot\psi(O;\theta)\) is continuous in \(\theta\) and admits a constant envelope on \(K\). Dominated convergence therefore permits differentiation under the expectation, proving \eqref{eq:diff_under} and showing that \(\dot\Psi_\bbP\) is continuous. Hence \(\Psi_\bbP\) is continuously differentiable.
\end{proof}
The following result is well known
\begin{lem}\label{lem:unicity}
Under assumptions \ref{ass:glm1} and \ref{ass:glm3}--\ref{ass:glm6},
\(\theta_0\) is the unique zero of \(\Psi_\bbP\).
\end{lem}

\begin{proof}
We already know that \(\Psi_\bbP(\theta_0)=0\). Let
\(\theta_\ast=(\eta_\ast,h^0_\ast,h^1_\ast)\) be any zero of
\(\Psi_\bbP\). We first identify the nuisance parameter. Let \(k\in\mathcal K\) and set \(\delta^k=\eta_0^k-\eta_\ast^k\). From equation
\eqref{eq:pooled-glm-sequential-assumption}, we obtain
\begin{align*}
& (\delta^k)^\transpose
  \bbE\!\left[\psi^k(O;\eta_\ast^k)\right] \\
&\quad=\bbE\!\left[
  \sum_s V_s^kD_s^k b_k\!\left(z_s^k(\eta_\ast^k)\right)
  \{g_k(z_s^k(\eta_0^k))-g_k(z_s^k(\eta_\ast^k))\}
  \{z_s^k(\eta_0^k)-z_s^k(\eta_\ast^k)\}
  \right].
\end{align*}

Each summand on the right-hand side is nonnegative. Indeed, by
\ref{ass:glm4}, \(b_k\) has the same sign as \(g_k^\prime\), while
\[
(g_k(x)-g_k(y)) (x-y)
\]
also has the sign of \(g_k^\prime\) whenever \(x\ne y\). Consequently,
\[
b_k(y) (g_k(x)-g_k(y))(x-y)>0
\qquad\text{whenever }x\ne y.
\]
 This implies that
\begin{equation*}
(\delta^k)^\transpose
  \bbE\!\left[\psi^k(O;\eta_\ast^k)\right]=0 \quad \textnormal{if and only if} \quad \bbE\!\left[
  \sum_s V_s^kD_s^k \{z_s^k(\eta_0^k)-z_s^k(\eta_\ast^k)\}^2
  \right]=0.
\end{equation*}
Moreover, if
\(\delta^k\neq 0\), then \ref{ass:glm5} yields
\[
\bbE\!\left[\sum_s V_s^kD_s^k\{(\delta^k)^\transpose Z_s^k\}^2\right]
=(\delta^k)^\transpose
  \bbE\!\left[\sum_sV_s^kD_s^kZ_s^k(Z_s^k)^\transpose\right]
  \delta^k
>0.
\]
This implies
\((\delta^k)^\transpose\bbE[\psi^k(O;\eta_\ast^k)]>0\),
contradicting the fact that \(\theta_\ast\) is a zero. Hence
\(\eta_\ast^k=\eta_0^k\) for every \(k\in\mathcal K\), and therefore
\(\eta_\ast=\eta_0\).

It remains to identify the hazard parameters. For every
\(j\in\{0,1\}\) and \(s=1,\ldots,\tau\), the corresponding population
equation, now evaluated at \(\eta_0\), satisfies
\begin{align*}
0
&=\bbE\!\left[W_{s-1}^j(\eta_0)R_s
  (E_s-h_{\ast,s}^j)\right]\\
&=(h_{0,s}^j-h_{\ast,s}^j)
  \bbE\!\left[W_{s-1}^j(\eta_0)R_s\right],
\end{align*}
where the defining equation for \(h_{0,s}^j\) gives the second equality.
The expectation is strictly positive by \ref{ass:glm6}, so
\(h_{\ast,s}^j=h_{0,s}^j\). Consequently, \(\theta_\ast=\theta_0\), which
proves uniqueness.
\end{proof}

\begin{lem}
Under assumption \ref{ass:glm1} to \ref{ass:glm6}, conditions (Z1), (Z2) and (Z4)--(Z6) are satisfied. If $\hat\eta_n$ satisfies
\eqref{eq:pooled-glm-approx-root} and the hazard estimators solve their weighted estimating equations exactly, then the estimator $\hat\theta_n$ also satisfies the approximate-root condition (Z3).
\end{lem}

\begin{proof}
\begin{enumerate}[label=(Z\arabic*)]
    \item  The map $\Psi_\bbP$ is continuous by Lemma \ref{lem:continuity}. The compactness of $\Theta$ \ref{ass:glm2} together with the uniqueness of the root of $\Psi_\bbP$ granted by Lemma \ref{lem:unicity} imply (Z1).
    \item  The patient trajectories are i.i.d. Moreover, the parameter space $\Theta$ is compact and the application $\theta \mapsto \psi_\theta$ defined in \eqref{eq:critical_curve} is continuous and bounded by a constant. Hence, the result follows from \cite[Example 19.8]{vanderVaart1998}
    \item For each $j$ and $s$, the weighted hazard estimator satisfies
\[
\bbE_{\bbP_n}\left[
W^j_{s-1}(\hat\eta_n)R_s(E_s-\hat h_{n,s}^j)
\right]=0.
\]
Consequently,
\[
\Psi_n(\hat\theta_n)
=
\begin{pmatrix}
\bbE_{\bbP_n}\psi^{\rm nuis}(O;\hat\eta_n)\\
0_{2\tau}
\end{pmatrix}.
\]
Condition \eqref{eq:pooled-glm-approx-root} therefore gives
$\|\Psi_n(\hat\theta_n)\|=o_{\bbP}(n^{-1/2})$, establishing (Z3).
    \item Assumption (Z4) is established in the proof of Lemma \ref{lem:continuity}
    \item  Assumption (Z5) is established in  the proof of Lemma \ref{lem:continuity}
    \item By Lemma \ref{lem:continuity}, we already know that  $A=\dot\Psi_\bbP(\theta_0)=\expect{\dot\psi(O;\theta_0)}$. It remains to prove that it is non-singular. Since  the nuisance scores do not depend on $h^0$ or $h^1$ ,the matrix $A$ is lower block triangular:
\begin{equation*}
    A=
    \begin{pmatrix}
        A_\eta & 0 \\
        B & D_h
    \end{pmatrix}
\end{equation*}
Hence, to check its invertibility, it is sufficient to check the invertibility its diagonal blocks $A_\eta$ and $D_h$. We further remark that since each nuisance score in equation \eqref{eq:pooled-glm-stacked-score} depends only
on its own parameter block, $A_\eta$ is block diagonal i.e.
\[
A_\eta
=\operatorname{diag}(A_{\Omega_0},A_{\Omega_1},A_C).
\]

We start by computing $A_k$ for $k \in \mathcal{K}$. Define the scalar
derivative
\begin{equation*}
\dot b_k(z):=\frac{d}{dz}b_k(z)
=\frac{g_k''(z)}{g_k(z)\{1-g_k(z)\}}
-\frac{\{g_k'(z)\}^2\{1-2g_k(z)\}}
       {g_k(z)^2\{1-g_k(z)\}^2}.
\end{equation*}
Hence, the Jacobian of $\psi^k(O_i;\eta^k)$ is
\begin{align*}
\dot\psi^k(O_i;\eta^k)
&:=\frac{\partial\psi^k(O_i;\eta^k)}
        {\partial(\eta^k)^\transpose}\\
&=\sum_s V^k_{i,s}D^k_{i,s}
\left[
\dot b_k\!\left(z^k_{i,s}(\eta^k)\right)
\{U^k_{i,s}-\mu^k_{i,s}(\eta^k)\} \right.\\
& \left. -b_k\!\left(z^k_{i,s}(\eta^k)\right)
g_k'\!\left(z^k_{i,s}(\eta^k)\right)
\right]
Z^k_{i,s}(Z^k_{i,s})^\transpose.
\end{align*}
Taking the expectation at $\eta_0$ yields
\begin{equation*}
A_k
:=\bbE\left [\dot\psi^k(O;\eta_0^k)\right ]
=-\bbE\left[
\sum_s V^k_sD^k_{s}\omega^k_s Z^k_s(Z^k_s)^\transpose
\right],
\qquad
\omega^k_s
:=\frac{\{g_k'(z^k_s(\eta_0^k))\}^2}
        {\mu^k_s(\eta_0^k)\{1-\mu^k_s(\eta_0^k)\}}.
\end{equation*}
It follows from \ref{ass:glm5} that \(-A_k\)  is positive definite, hence invertible.

We now compute $D_h$.

\[
\frac{\partial}{\partial h^j_s}
\psi_{j,s}(O;\eta,h^j_s)
=-W^j_{s-1}(\eta)R_s.
\]
Hence,
\begin{equation*}
D_h
:=-\operatorname{diag}\left(
\bbE[W^j_{s-1}(\eta_0)R_s]:j\in\{0,1\},\ s=1,\ldots,\tau
\right),
\end{equation*}
and it follows from \ref{ass:glm6} that
\[
\bbE[W^j_{s-1}(\eta_0)R_s]>0
\qquad\text{for every }j\in\{0,1\},\ s=1,\ldots,\tau.
\]
which implies that $D_h$ is invertible. Thus $A$ is invertible
\end{enumerate}
\end{proof}

\section{Appendix to Section \ref{sec:simulation}} \label{app:simulation}

We generate the data of Section~\ref{sec:simulation} as follows. Baseline covariates are generated as
\[
M_1 \sim \mathcal N(0,1),
\qquad
M_2 \sim \mathrm{Bernoulli}(0.5).
\]
At visit \(t=0\),
\[
X_{1,0}=0.2\,M_1+\varepsilon_{1,0},
\qquad
\varepsilon_{1,0}\sim\mathcal N(0,1),
\]
and
\[
X_{2,0}\sim \mathrm{Bernoulli}\!\left(\mathrm{expit}(-0.2+0.4\,M_2)\right).
\]
For \(t\ge1\),
\[
X_{1,t}
=
0.3+0.7\,X_{1,t-1}-0.5\,A_{t-1}-0.25\,A_{t-2}+\varepsilon_{1,t},
\]
with \(\varepsilon_{1,t}\sim\mathcal N(0,1)\), and
\[
X_{2,t}
\sim
\mathrm{Bernoulli}\!\left(
\mathrm{expit}(-0.2+1.0\,X_{2,t-1}-0.5\,A_{t-1}-0.25\,A_{t-2})
\right).
\]
Treatment is generated according to
\begin{align*}
\Pr(A_t=1 \mid M,X_t,A_{t-1},A_{t-2})
&=
\mathrm{expit}\!\Big(
\alpha_k
+
0.45\times 0.5\,
\bigl(
0.08\,M_1 + 0.12\,M_2 + 0.10\,X_{1,t} + 0.10\,X_{2,t}
\bigr) \\
&\qquad\qquad
+ 0.15\,A_{t-1}
+ 1.10\,A_{t-2}
\Big),
\end{align*}
with
\[
(\alpha_0,\ldots,\alpha_9)
=
(10.0,\ 1.0,\ -4.0,\ 0.3,\ 0.3,\ 0.3,\ 0.3,\ 0.3,\ -4.0,\ -4.0).
\]
The event process follows a piecewise exponential model. Writing
\[
|A_t|=\sum_{s=0}^t A_s,
\]
the linear predictor is
\[
\eta_{Y,t}
=
-4.2
+
0.25\,M_1
+
0.25\,M_2
-
0.45\,X_{1,t}
-
0.45\,X_{2,t}
+
A_t\bigl(-1.5+0.3\,|A_t|\bigr).
\]
For \(t=0,\ldots,8\), the event probability over interval \([t,t+1)\) is
\[
p_{Y,t}=1-\exp\bigl(-\exp(\eta_{Y,t})\bigr),
\]
whereas on the last interval \([9,\tau]\), with \(\tau=12\), the hazard is constant with rate
\[
\lambda_{Y,9}=\exp(\eta_{Y,9}).
\]
Natural censoring is generated independently from an exponential distribution with rate
\[
\lambda_C=\exp(-4.6).
\]



\end{document}